\documentclass[amssymb,amsfonts]{amsart}
\usepackage{amssymb, latexsym}

\newtheorem{thm}{Theorem}

\newtheorem{prop}[thm]{Proposition}

\newtheorem{lemma}[thm]{Lemma}
\newtheorem{cor}[thm]{Corollary}
\newtheorem{rmk}[thm]{Remark}

\newcommand{\lbar}{{\underline{l}}}
\newcommand{\mbar}{{\bar{m}}}
\newcommand{\ubar}[1]{{\underline{#1}}}

\begin{document}
\title{A space-time characterization of the Kerr-Newman metric}
\author[W. W. Wong]{Willie W. Wong}
\address{408 Fine Hall\\ Princeton University\\ Princeton, NJ}
\email{wwong@math.princeton.edu}

\begin{abstract}
In the present paper, the characterization of the Kerr metric found by
Marc Mars is extended to the Kerr-Newman family. A simultaneous
alignment of the Maxwell field, the Ernst two-form of the
pseudo-stationary Killing vector field, and the Weyl curvature of the
metric is shown to imply that the
space-time is locally isometric to domains in the Kerr-Newman metric.
The paper also presents an extension of Ionescu and Klainerman's null
tetrad formalism to explicitly include Ricci curvature terms. 
\end{abstract}

\maketitle

\section{Introduction}
There are relatively few known exact solutions, which have metrics
that can be easily written down in closed form, to the Einstein
equations in the asymptotically flat case. Among the most well-known
of such solutions are the Kerr family \cite{Kerr63} of axially-symmetric, stationary
vacuum space-times, which represent the exterior space-time of a
spinning massive object, and the Kerr-Newman family \cite{Newman65} of
axially-symmetric, stationary, electrovac space-times, which represent
the exterior space-time of a spinning, electrically charged, massive
object. A natural question to ask about special solutions such as
these is whether they are stable or unique, where stability or
uniqueness is chosen among some
suitable class. While much progress had been made toward the
uniqueness problem, less can be said about the stability problem.  

It should not be surprising that the first results of this kind came
in the context of the special static, spherically symmetric members of
the Kerr and Kerr-Newman families: the Schwarzschild and
Reissner-Nordstr\"om solutions respectively. It has been known since
the 1920's \cite{JoRa05} that the Schwarzschild family completely
parametrizes the spherically symmetric solutions to Einstein's vacuum
equations; a similar result is later obtained for the
Reissner-Nordstr\"om family for spherically symmetric solutions of the
electrovac equations. These results now go under the name of Birkhoff's
theorem. In particular, Birkhoff's theorem essentially states that
spherical symmetry implies staticity and asymptotic flatness of the 
space-time. The next step forward came in the 1960's, when Werner Israel 
established \cite{Israel67,Israel68} what is, loosely speaking, the
converse of Birkhoff's theorem: a static, asymptotically flat
space-time that is regular on the event horizon must be spherically
symmetric. Brandon Carter's 1973 Les Houches report \cite{CarterLH}
finally sparked an attempt to similarly characterize the Kerr and
Kerr-Newman families: he showed that asymptotically flat, stationary,
and axially-symmetric
solutions to the vacuum (electrovac) equations form a two-parameter (three-)
family. Between D.~C.~Robinson \cite{Robinson75}, P.~O.~Mazur \cite{Mazur82}, and
G.~L.~Bunting \cite{Bunting83}, Carter's program was completed and the
Kerr and Kerr-Newman families are established as essentially the
unique solutions to the asymptotically flat, stationary,
axially-symmetric Einstein's equations. 

A different approach was taken by Walter Simon \cite{Simon84a,
Simon84b} to study the characterization of the Kerr and Kerr-Newman
families among stationary solutions. He constructed three-index
tensors that are, heuristically speaking, complexified versions of the
Cotton tensors on the stationary spatial slices (to be more precise,
the manifold of trajectories generated by the time-like Killing vector
field). By considering the
multiple moments of a stationary, asymptotically flat end, Simon
showed that the vanishing of the three-index tensor is equivalent to
the multiple moments being equal to those of the Kerr and Kerr-Newman
families. Simon's work was
later extended by Marc Mars \cite{Mars99} to the construction of the
so-called Mars-Simon tensor, which is a four-index tensor constructed
relative to space-time quantities, as opposed to Simon's original
construction relative to the induced metric on the spatial slices. As
was shown by Mars, the vanishing of the Mars-Simon tensor indicates
an alignment of the principal null directions of the Ernst two-form
(for definition see Section \ref{SectSetup}) with those of the Weyl
curvature tensor, with the particular proportionality factor allowing
one to write down the local form of the metric explicitly and verify
that the space-time is locally isometric to the Kerr space-time. 

The method employed by Marc Mars and the present paper bears much
similarity to the work of R.~Debever, N.~Kamran, and R.~G.~McLenaghan
\cite{DebeverKamranMcLenaghan83}, in which the authors assumed (i) the
space-time is of Petrov type $D$, (ii) the principal null directions of the
Maxwell tensor align (nonsingularly) with that of the Weyl tensor,
(iii) a technical hypothesis to allow the use of the generalized
Goldberg-Sachs theorem (see Chapter 7 in \cite{ExactSolutions} for
example and references), and integrated the Newman-Penrose variables to
arrive at explicit local forms of the metric in terms of several
constants that can be freely specified. In view of the work of Debever
et al., the assumptions taken in this paper merely guarantees that
their hypotheses (i) and (ii) hold, and that (iii) becomes ancillary to
a stronger condition derived herein that circumvents the
Goldberg-Sachs theorem as well as prescribes definite values for all
but three (mass, angular momentum, and charge) of the free constants.

In the current work, we extend the construction of Mars to define a
four-tensor analogous to the Mars-Simon tensor \emph{and}, in addition, a two-form such that
their simultaneous vanishing guarantees the simultaneous
alignment of the principal null directions of the Ernst two-form, the
Maxwell field, and the Weyl tensor, with proportionality factors that
allow us to write down the local form of the metric and demonstrate a
local isometry to Kerr-Newman space-time. It is worth mentioning the
work of Donato Bini et al.~\cite{Bini04} in which they keep the same
definition of the Mars-Simon tensor, while modifying the definition of
the Simon three-tensor with a source term that corresponds to the
stress-energy tensor associated to the electromagnetic field. They
were then able to show that the vanishing of the modified Simon tensor
implies also the alignment of principal null directions. In the
present work, we absorb the source term into the Mars-Simon tensor
itself using only space-time quantities by sacrificing a need for
an auxiliary two-form, thus we are able to argue in much of the same
way as Mars \cite{Mars99} an explicit computation for the metric expressed
in local co\"ordinates, thereby giving a characterization of the
Kerr-Newman space-time.
 
In a forthcoming paper we hope to use this characterization, combined
with the Carleman estimate techniques of Ionescu and Klainerman
\cite{IoKl07a, IoKl07b} to obtain an analogous uniqueness result for
smooth stationary charged black holes. 

We should note that the characterization found here is
essentially local, analogous to Theorem 1 in \cite{Mars00} (see Theorem
\ref{ThmMain} below). The global feature of the space-time, namely
asymptotic flatness, is only used to a priori prescribe the values of
certain constants using the mass and charge at infinity (compare
Corollary \ref{CorMain} below). 
The present paper is organized as follows: In Section \ref{SectSetup},
we first review the concept of complex anti-self-dual two-forms and
their properties, then we present the basic assumptions on the
space-time under consideration, followed by a quick review of Killing
vector fields, and conclude with the principal definitions and a
statement of the main theorem and corollary. In Section \ref{SectMainProof}, we
demonstrate the technique of the proof for the main theorem through explicit
construction of a local isometry, using the tools of the null tetrad
formalism of Ionescu and Klainerman \cite{IoKl07a}. In Section
\ref{SectGlobalRes} we prove the corollary. We also include
an appendix extending the framework established by Ionescu and
Klainerman to explicitly include terms coming from Ricci curvature
(terms which were not necessary in \cite{IoKl07a,IoKl07b} since they
consider vacuum Einstein metrics), and including a dictionary between
the coefficients in this formalism and those of the Newman-Penrose
system. 

The author would like to thank his thesis advisor, Sergiu Klainerman,
for pointing him to this problem; and Pin Yu, for many valuable
discussions. This manuscript also owes much to the detailed readings
and suggestions by the anonymous referee. The research for this work
was performed while the author was supported by an NSF Graduate Research
Fellowship. 

\section{Set-up and definitions}\label{SectSetup}

\subsection{Complex anti-self-dual two-forms}
On a four dimensional Lorentzian space-time $(\mathcal{M},g_{ab})$, 
the Hodge-star operator
$*:\Lambda^2T^*\mathcal{M}\to\Lambda^2T^*\mathcal{M}$ is a linear transformation on the
space of two-forms. In index notation, 
\[ {}^*X_{ab} = \frac{1}{2}\epsilon_{abcd}X^{cd} \]
where $\epsilon_{abcd}$ is the volume form and index-raising is done
relative to the metric $g$. Since we take the metric signature to be
$(-,+,+,+)$, we have that $** = - \mathop{Id}$, which introduces a
complex structure on the space $\Lambda^2T^*\mathcal{M}$. By complexifying and
extending the action of $*$ by linearity, we
can split $\Lambda^2T^*\mathcal{M}\otimes_\mathbb{R}\mathbb{C}$ into the
eigenspaces $\Lambda_{\pm}$ of $*$ with eigenvalues $\pm i$. We say that 
an element of 
$\Lambda^2T^*\mathcal{M}\otimes_\mathbb{R}\mathbb{C}$ is \emph{self-dual} if it
is an eigenvector of $*$ with eigenvalue $i$, and we say that it is
\emph{anti-self-dual} if it has eigenvalue $-i$. It is easy to check that
given a real-valued two-form $X_{ab}$, the two form
\begin{equation}\label{EqDefASDPart}
\mathcal{X}_{ab} := \frac{1}{2}(X_{ab} + i {}^*X_{ab})
\end{equation}
is anti-self-dual, while its complex conjugate
$\bar{\mathcal{X}}_{ab}$ is self-dual. 

In the sequel we shall, in general, write elements of
$\Lambda^2T^*\mathcal{M}$
with upper-case Roman letters, and their corresponding anti-self-dual
forms with upper-case calligraphic letters. The projection
\[ X_{ab} = \mathcal{X}_{ab} + \bar{\mathcal{X}}_{ab} \]
is a natural consequence of (\ref{EqDefASDPart}). 

Here we record some product properties \cite{Mars99} of two-forms:
\begin{subequations}
\begin{align} 
X_{ac}Y_b{}^c - {}^*X_{ac}{}^*Y_b{}^c &= \frac12 g_{ab} X_{cd}Y^{cd}\\
X_{ac}{}^*X_b{}^c &=\frac14 g_{ab}X_{cd}{}^*X^{cd} \\
\mathcal{X}_{ac}\mathcal{Y}_b{}^c + \mathcal{Y}_{ac}\mathcal{X}_b{}^c
& = \frac12g_{ab}\mathcal{X}_{cd}\mathcal{Y}^{cd} \\
\label{EqSelfDualSelfProduct} \mathcal{X}_{ac}\mathcal{X}_b{}^c &=
\frac14g_{ab}\mathcal{X}_{cd}\mathcal{X}^{cd} \\
\mathcal{X}_{ac}X_b{}^c - \mathcal{X}_{bc}X_a{}^c &=0\\
\mathcal{X}_{ab}Y^{ab} &= \mathcal{X}_{ab}\mathcal{Y}^{ab}\\
\mathcal{X}_{ab}\bar{\mathcal{Y}}^{ab} &= 0
\end{align}
\end{subequations}

Now, the projection operator
$\mathcal{P}_\pm:\Lambda^2T^*\mathcal{M}\otimes_\mathbb{R}\mathbb{C}\to \Lambda_{\pm}$  can be given in
index notation as 
\begin{align*}
(\mathcal{P}_+X)_{ab} & =
\bar{\mathcal{I}}_{abcd}X^{cd} \\
(\mathcal{P}_-X)_{ab} & = \mathcal{I}_{abcd}X^{cd}\\
\mbox{where}~ \mathcal{I}_{abcd} & = \frac14(g_{ac}g_{bd} -
g_{ad}g_{bc} + i \epsilon_{abcd})
\end{align*}
With the complex tensor $\mathcal{I}_{abcd}$, we can define with the
notation 
\begin{equation}
(\mathcal{X}\tilde\otimes\mathcal{Y})_{abcd} :=
\frac12\mathcal{X}_{ab}\mathcal{Y}_{cd} +
\frac12\mathcal{Y}_{ab}\mathcal{X}_{cd} -
\frac13\mathcal{I}_{abcd}\mathcal{X}_{ef}\mathcal{Y}^{ef}
\end{equation}
a symmetric bilinear product taking two anti-self-dual forms to a
complex $(0,4)$-tensor. It is simple to verify that such a tensor
automatically satisfies the algebraic symmetries of the Weyl conformal
tensor: i) it is antisymmetric in its first two, and last two, indices
$(\mathcal{X}\tilde\otimes\mathcal{Y})_{abcd} = -
(\mathcal{X}\tilde\otimes\mathcal{Y})_{bacd} = -
(\mathcal{X}\tilde\otimes\mathcal{Y})_{abdc}$
ii) it is symmetric swapping the first two and the last two sets of
indices $(\mathcal{X}\tilde\otimes\mathcal{Y})_{abcd} =
(\mathcal{X}\tilde\otimes\mathcal{Y})_{cdab}$ iii) it verifies the
first Bianchi identity $(\mathcal{X}\tilde\otimes\mathcal{Y})_{abcd} +
(\mathcal{X}\tilde\otimes\mathcal{Y})_{bcad} +
(\mathcal{X}\tilde\otimes\mathcal{Y})_{cabd} = 0$ and iv) it is
trace-free $(\mathcal{X}\tilde\otimes\mathcal{Y})_{abcd}g^{ac} = 0$.
For lack of a better name, this product will be referred to as a
\emph{symmetric spinor product}, using the fact that if we represent
in spinor co\"ordinates $\mathcal{X}_{ab} = f_{AB}\epsilon_{A'B'}$ and
$\mathcal{Y}_{ab} = h_{AB}\epsilon_{A'B'}$ (where $f_{AB} = f_{BA}$,
and similarly for $h_{AB}$), the product can be written
as
\[ (\mathcal{X}\tilde\otimes\mathcal{Y})_{abcd} =
f_{(AB}h_{CD)}\epsilon_{A'B'}\epsilon_{C'D'} \]
where $(\cdot)$ denotes complete symmetrization of the indices. 

\subsection{The basic assumptions on the space-time and some
notational definitions}

We consider a space-time $(\mathcal{M},g_{ab})$ and a Maxwell two-form
$H_{ab}$ on $\mathcal{M}$ satisfying the following basic assumptions
\begin{itemize}
\item[(A1)] $\mathcal{M}$ is a four-dimensional, orientable,
paracompact, simply-connected manifold.
\item[(A2)] $g_{ab}$ is a smooth Lorentzian metric on $\mathcal{M}$.
$H_{ab}$ is a smooth two-form.
\item[(A3)] The metric $g_{ab}$ and the
Maxwell form $H_{ab}$ satisfy the Einstein-Maxwell field equations.
Namely
\begin{align*}
R_{ab} &= T_{ab}\\
\nabla_{[c}H_{ab]} &= 0\\
\nabla^aH_{ab} & = 0
\end{align*} 
where $T_{ab} = 2 H_{ac}H_b{}^c - \frac12 g_{ab}H_{cd}H^{cd} =
4\mathcal{H}_{ac}\bar{\mathcal{H}}_b{}^c$ is the rescaled
stress-energy tensor, which is traceless and divergence free by
construction, and square brackets $[\cdot]$ around indices
means full anti-symmetrization. 
\item[(A4)]  $(\mathcal{M},g_{ab})$ admits a non-trivial smooth Killing
vector field $t^a$, 
and the Maxwell field  $H_{ab}$ inherits the
Killing symmetry, i.e.~its Lie derivative $\mathcal{L}_tH_{ab} = 0$.
\end{itemize}
In the sequel we will state a local and a global version of the result. For the
local theorem, we need to assume
\begin{itemize}
\item[(L)] the Killing vector field $t^a$ is time like somewhere on
the space-time $(\mathcal{M},g_{ab})$, and $H_{ab}$ is non-null on
$\mathcal{M}$. (In other words,
writing the anti-self-dual part $\mathcal{H}_{ab} = \frac12(H_{ab}+ i
{}^*H_{ab})$, we require $\mathcal{H}_{ab}\mathcal{H}^{ab} \neq 0$
everywhere on $\mathcal{M}$.)
\end{itemize}
And for the global result, we assume
\begin{itemize}
\item[(G)] that $(\mathcal{M},g_{ab})$ contains a stationary
asymptotically flat end $\mathcal{M}^\infty$ where $t^a$ tends to a time
translation at infinity, with the Komar mass $M$ of $t^a$ non-zero in
$\mathcal{M}^\infty$. We also assume the total charge $q = \sqrt{q_E^2 +
q_B^2}$ of the Maxwell field, where $q_E$ and $q_B$ denote the
electric and magnetic charges, is non-zero in $\mathcal{M}^\infty$.
\end{itemize}

\begin{rmk}We quickly recall the definition of stationary
asymptotically flat end: $\mathcal{M}^\infty$ is an open submanifold
of $\mathcal{M}$
diffeomorphic to $(t_0,t_1)\times (\mathbb{R}^3\setminus \bar{B}(R))$ with
the metric stationary in the $t$ variable, $\partial_t g_{ab} = 0$,
and satisfying the decay condition
\[ |g_{ab} - \eta_{ab}| + |r\partial g_{ab}| \leq C r^{-1} \]
for some constant $C$; $r$ is the radial co\"ordinate on
$\mathbb{R}^3$ and $\eta$ is the Minkowski metric. In addition, we
will also require a decay condition for the Maxwell field
\[ |H_{ab}| + |r\partial H_{ab}| \leq C' r^{-2} \]
for some constant $C'$. 
\end{rmk}

We record here some notational definitions: $R_{abcd}$ is the Riemann
curvature tensor, with the standard decomposition
\[ R_{abcd} = W_{abcd} + \frac12 (R_{ac}g_{bd} + R_{bd}g_{ac} -
R_{ad}g_{bc}- R_{bc}g_{ad}) - \frac16R(g_{ac}g_{bd} - g_{ad}g_{bc}) \]
where $W_{abcd}$ is the conformal (Weyl) curvature tensor, 
$R_{ac} = R_{abcd}g^{bd}$ the Ricci curvature tensor, and $R$
the scalar curvature. For the electro-vac system, this reduces to
\[ R_{abcd} = W_{abcd} + \frac12(T_{ac}g_{bd} + T_{bd}g_{ac} -
T_{ad}g_{bc} - T_{bc}g_{ad}) \]
For a $(0,4)$ tensor $K_{abcd}$ satisfying the algebraic symmetries of the
Riemann tensor, we define the left- and right-duals 
\begin{align*}
{}^*K_{abcd} & = \frac12\epsilon_{abef}K^{ef}{}_{cd} \\
K^*{}_{abcd} & = \frac12K_{ab}{}^{ef}\epsilon_{efcd}
\end{align*}
In general the left- and right-duals are \emph{not} equal. If, in
addition,
$K_{abcd}$ is also trace-free (i.e., is a Weyl field in the sense
defined in \cite{ChKl93}), a simple calculation shows
that the left- and right-duals are equal. Therefore we can define the
anti-self-dual complex Weyl curvature tensor
\begin{equation*}
\mathcal{C}_{abcd} = \frac12 (W_{abcd} + i{}^*W_{abcd})
\end{equation*}
It may be of independent interest to note that in the electro-vac case
\[ \mathcal{C}_{abcd} = (\mathcal{P}_-R\mathcal{P}_-)_{abcd} =
\mathcal{I}_{abef}R^{efgh}\mathcal{I}_{ghcd} \]
(when the scalar curvature $R\neq 0$, it also presents a contribution
to this projection).
In other words, treating the Riemann curvature tensor as a map from
$\Lambda^2T^*\mathcal{M}\otimes_{\mathbb{R}}\mathbb{C}$ to itself, the Weyl curvature takes
$\Lambda_+\to\Lambda_+$ and $\Lambda_-\to\Lambda_-$, whereas the
Kulkani-Nomizu product of Ricci curvature with the metric induces a
intertwining map that takes $\Lambda_-\to\Lambda_+$ and vice versa. 

Lastly, we define the following notational shorthand for
Lorentzian ``norms'' of tensor fields. For an arbitrary
$(j,k)$-tensor $Z^{a_1a_2\ldots a_j}_{b_1b_2\ldots b_k}$, we write
\[ Z^2 = g_{a_1a_1'}g_{a_2a_2'}\cdots g_{a_ja_j'}g^{b_1b_1'}\cdots
g_{b_kb_k'} Z^{a_1a_2\ldots a_j}_{b_1b_2\ldots b_k} Z^{a_1'a_2'\ldots
a_j'}_{b_1'b_2'\ldots b_k'} \]
for the inner-product of $Z^{\cdots}_{\cdots}$ with itself. Note that in
the semi-Riemannian setting, $Z^2$ can take arbitrary sign.

\subsection{The Killing symmetry}

Given $(\mathcal{M},g_{ab})$ a smooth, four-dimensional Lorentzian
manifold, and assuming that it admits a smooth Killing vector field
$t^a$, we can define the Ernst two-form 
\begin{equation}
F_{ab} = \nabla_at_b - \nabla_bt_a = 2\nabla_at_b
\end{equation}
the second equality a consequence of the Killing equation. As is
well-known, the Ernst two-form satisfy
\begin{equation}
\nabla_cF_{ab} = 2 \nabla_c\nabla_at_b = 2R_{dcab}t^d
\end{equation}
This directly implies a divergence-curl system (in other words, a
Maxwell equation with source terms) satisfied by the
two-form
\begin{align*}
\nabla_{[c}F_{ab]} & = 0\\
\nabla^aF_{ab} &= -2 R_{db}t^d
\end{align*}
Here we encounter one of our primary differences from \cite{Mars99}:
a space-time satisfying the Einstein vacuum equations is Ricci-flat,
and the above implies that the Ernst two-form satisfies the
sourceless Maxwell equations. In particular, for the vacuum case, we
have
\[ \nabla_{[c}\mathcal{F}_{ab]} = 0\]
and a calculation then verifies that 
\[ \nabla_{[c}(\mathcal{F}_{a]b}t^b) = 0~. \]
Thus from simple-connectivity, an Ernst potential $\sigma$ is
constructed for
\[ \nabla_a\sigma = \mathcal{F}_{ab}t^b~. \]

In the non-vacuum case that this paper deals with, this construction
cannot be exactly carried through. However, the essence of the
construction above is the following fact disjoint from the
semi-Riemannian structure of our setup: consider a smooth manifold
$\mathcal{M}$, a smooth differential form $X$, and a smooth 
vector-field $v$. We have the defining relation
\[ \mathcal{L}_vX = i_v\circ dX + d\circ i_vX \]
where $\mathcal{L}_v$ stands for the Lie derivative relative to the
vector-field $v$, and $i_v$ is the interior derivative. Thus if $X$ 
is a closed form, and $v$ is a symmetry of $X$ (i.e. 
$\mathcal{L}_vX = 0$), we must have $i_vX$ is closed also. 

Applying to the Einstein-Maxwell equations, we take $X$ to be the
anti-self-dual Maxwell form $\mathcal{H}_{ab}$, which by Maxwell's
equations is closed. The vector-field $v$ is naturally the Killing
field $t^a$, so we conclude that the complex-valued one-form
$\mathcal{H}_{ab}t^a$ is closed, and since $\mathcal{M}$ is taken to
be simply connected, also exact. In the sequel we will use the
complex-valued function $\Xi$, which is defined by
\begin{equation}\label{EqDefXi}
\nabla_b\Xi = \mathcal{H}_{ab}t^a~.
\end{equation}
Notice that \emph{a priori} $\Xi$ is only defined up to the addition
of a constant. In the global case (making the assumption (G)), we can
use the asymptotic decay of the Maxwell field to require that
$\Xi\to 0$ at spatial infinity and fix $\Xi$ uniquely. The 
function $\Xi$ takes the 
place of the Ernst potential $\sigma$ used in \cite{Mars99}. 

Lastly, we record here two calculations used in the sequel: first we
write down explicitly the derivative of $\mathcal{F}_{ab}$
\begin{align}
\nonumber \nabla_c\mathcal{F}_{ab} &= (R_{dcab} + i R^*_{dcab})t^d\\
\label{EqDerASDErnst} & = 2\mathcal{C}_{dcab}t^d +
\frac12(T_{ad}g_{bc}+T_{bc}g_{ad}-T_{ac}g_{bd}-T_{bd}g_{ac})t^d \\
\nonumber & \qquad +
\frac{i}{2}(T_d{}^e\epsilon_{ecab} + T_c{}^f\epsilon_{dfab})t^d
\end{align}
we will also need the following fact about Killing vector fields.
Consider the product ${}^*F_{ab}{}^*F_{cd} = \frac14
\epsilon_{abef}\epsilon_{cdgh}F^{ef}F^{gh}$. We can expand the product
of the Levi-Civita symbol/volume form in terms of the metric:
\[ \epsilon_{ijkl}\epsilon^{qrst} = -24 g_i^{[q}g_j^rg_k^sg_l^{t]} \]
By explicit computation using this expansion, we arrive at the fact
\begin{align*}
{}^*F_{mx}t^x{}^*F_{ny}t^y & = \frac12
F_{ab}F^{ab} (t_mt_n - t_xt^x g_{mn}) + g_{mn} F_{xa}t^xF^{ya}t_y -
F_{nx}t^xF_{my}t^y \\
& \qquad + F^{bx}t_xt_mF_{nb} + F^{bx}t_xt_nF_{mb} + t_xt^x
F_{ma}F_n{}^a
\end{align*}
Writing $t^2 = t_at^a$, we use the fact $\nabla_bt^2 = t^aF_{ba}$ and
obtain equation (13) from \cite{Mars99}:
\begin{align}
\label{EqFStarFStar} {}^*F_{mx}t^x{}^*F_{ny}t^y & = \frac12
F_{ab}F^{ab} ( t_mt_n - g_{mn}t^2 ) + g_{mn}\nabla_at^2\nabla^at^2 -
\nabla_mt^2\nabla_nt^2 \\
\nonumber & \qquad + t_mF_{nb}\nabla^bt^2 + t_nF_{mb}\nabla^bt^2 +
t^2F_{ma}F_n{}^a
\end{align}

\subsection{The Mars-Simon tensor for Kerr-Newman space-time;
statement of the main theorems}

We first state the main result of this paper, which establishes a
purely local characterization of the Kerr-Newman metric. This
formulation is comparable to that of Theorem 1 in \cite{Mars00}. The
conditions given below on the constants $C_2$ and $C_4$ are analogous
to the conditions for the constants $l$ and $c$ in the aforementioned
theorem.
\begin{thm}[Main Local Theorem]\label{ThmMain}
Assuming (A1)-(A4) and (L), and assuming that there exists a complex
scalar $P$, a normalization for $\Xi$, and a complex constant $C_1$
such that
\begin{enumerate}
\item $P^{-4} = - C_1^2 \mathcal{H}_{ab}\mathcal{H}^{ab}$
\item $\mathcal{F}_{ab} = 4\bar{\Xi}\mathcal{H}_{ab}$
\item $\mathcal{C}_{abcd} =
3P(\mathcal{F}\tilde\otimes\mathcal{H})_{abcd}$
\end{enumerate}
then we can conclude
\begin{enumerate}
\item there exists a complex constant $C_2$ such that $P^{-1} - 2\Xi =
C_2$;
\item there exists a real constant $C_4$ such that $t_at^a + 4|\Xi|^2
= C_4$.
\end{enumerate}
If $C_2$ further satisfies that $C_1\bar{C}_2$ is real, and that $C_4$
is such that $|C_2|^2 - C_4 = 1$, then we also have
\begin{enumerate}
\setcounter{enumi}{2}
\item $\mathfrak{A} = |C_1|^2 P\bar{P}
(\Im C_1\nabla P)^2 + (\Im C_1P)^2$ is a positive real constant on
the manifold\footnote{$\Im$
will
be used to denote the imaginary part of an expression. Notice
that $\mathfrak{A}$ is well defined even though $C_1$ can be replaced by $-C_1$.
One should observe the freedom to replace $C_1$ by $-C_1$ also
manifests in the remainder of this paper; it shall not be further
remarked upon.},
\item and 
$(\mathcal{M},g_{ab})$ is locally isometric to a Kerr-Newman
space-time of total charge $|C_1|$, angular momentum
$\sqrt{\mathfrak{A}}C_1\bar{C}_2$, and mass $C_1\bar{C}_2$.
\end{enumerate}
\end{thm}
The local theorem yields, via a simple argument, the following
characterization of the Kerr-Newman metric among stationary
asymptotically flat solutions to the Einstein-Maxwell system.
\begin{cor}[Main Global Result]\label{CorMain}
We assume (A1)-(A4) and (G), and let $q_E$, $q_B$, and $M$ be the
electric charge, magnetic charge, and Komar mass of the space-time at
one asymptotic end. We choose the normalization for $\Xi$ such that
it vanishes at spatial infinity. If we assume there exists a complex
function $P$ defined wherever $\mathcal{H}^2\neq 0$ such that
\begin{enumerate}
\item $P^{-4} = - (q_E + i q_B)^2 \mathcal{H}_{ab}\mathcal{H}^{ab}$
when $\mathcal{H}^2\neq 0$
\item $\mathcal{F}_{ab} = (4\bar{\Xi} - \frac{2M}{q_E + i
q_B})\mathcal{H}_{ab}$ everywhere
\item $\mathcal{C}_{abcd} =
3P(\mathcal{F}\tilde\otimes\mathcal{H})_{abcd}$ when $P$ is defined
\end{enumerate} 
then we can conclude that
\begin{enumerate}
\item $\mathcal{H}^2$ is non-vanishing globally,
\item $\mathfrak{A} = (q_E^2 + q_B^2) P\bar{P} (\Im (q_E + i q_B)\nabla P)^2 +
(\Im (q_E + i q_B) P)^2$ is a real-valued positive constant on the manifold,
\item and $(\mathcal{M},g_{ab})$ is everywhere locally isometric to a
Kerr-Newman space-time of total charge $q = \sqrt{q_E^2 + q_B^2}$,
angular momentum $\sqrt{\mathfrak{A}}M$, and mass $M$. 
\end{enumerate}
\end{cor}

For ease of notation, we write
the complex scalar $P$, the complex anti-self-dual form
$\mathcal{B}_{ab}$, and the complex anti-self-dual Weyl field
$\mathcal{Q}_{abcd}$ for the following expressions
\begin{subequations}
\begin{align}
\label{EqDefP} P^4 &:= - \frac{1}{C_1^2\mathcal{H}_{ab}\mathcal{H}^{ab}}\\
\label{EqDefB} \mathcal{B}_{ab} &:= \mathcal{F}_{ab} + (2\bar{C}_3
-4\bar{\Xi})\mathcal{H}_{ab} \\
\label{EqDefQ} \mathcal{Q}_{abcd} &:= \mathcal{C}_{abcd} -
3P(\mathcal{F}\tilde\otimes\mathcal{H})_{abcd}
\end{align}
\end{subequations}
By an abuse of language, in the sequel, the statement
``$\mathcal{B}_{ab} = 0$'' will be understood to mean the alignment
condition (2) in Theorem \ref{ThmMain} when we work under assumption
(L), or the alignment condition (2) in Corollary \ref{CorMain} when we
work under assumption (G), with suitably defined constants and
normalizations. Similarly, the statement ``$\mathcal{Q}_{abcd} = 0$''
will be taken to mean the existence of a suitable function $P$ such
that the appropriate alignment condition (3) is satisfied under
suitable conditions.

We end this section with a heuristic motivation of why the pair
$\mathcal{B}_{ab},\mathcal{Q}_{abcd}$ is a generalization of the
Mars-Simon tensor constructed in \cite{IoKl07a}. Assuming (G) and
suppose we have
$\mathcal{B}_{ab}$ and $\mathcal{Q}_{ab}$ both vanishing, and we take
the $q\to 0$ Kerr limit. Formally we define the quantity 
\[ \mathcal{G}_{ab} = -\frac{2M}{q_E + i q_B}\mathcal{H}_{ab} \]
when $q\neq 0$. The vanishing of $\mathcal{B}_{ab}$ becomes
\[ \mathcal{G}_{ab} = \mathcal{F}_{ab} / (1-\frac{2(q_E + i q_B)}{M}\bar{\Xi}) \]
and $P$ satisfies 
\[ P^4 = - \frac{4M^2}{(q_E + i q_B)^4\mathcal{G}_{ab}\mathcal{G}^{ab}} \]
Then we have
\[ 0 = \mathcal{Q}_{abcd} = \mathcal{C}_{abcd} +
\frac{3}{2M(-\frac{1}{4M^2}\mathcal{G}_{kl}\mathcal{G}^{kl})^{1/4}}(\mathcal{F}\tilde\otimes\mathcal{G})_{abcd}
\]
Now, formally taking $q\to 0$, we have that $\mathcal{B}_{ab} = 0 \to
\mathcal{G}_{ab} = \mathcal{F}_{ab}$, and 
\[ \mathcal{Q}_{abcd} = 0 \to \mathcal{C}_{abcd} = -
\frac{3}{(-4M^2\mathcal{F}_{kl}\mathcal{F}^{kl})^{1/4}}(\mathcal{F}\tilde\otimes\mathcal{F})_{abcd}
\]
which by inspection is the same vanishing condition imposed by the
Mars-Simon tensor in \cite{IoKl07a} or the vanishing condition in
Lemma 5 of \cite{Mars99} (the difference of a factor of 2 is due to a
factor of 2 difference in the definitions of anti-self-dual two-forms
and of the Ernst two-form). 

\section{Proof of the main local theorem}\label{SectMainProof}

Throughout this section we assume the statements (A1)-(A4) and (L). 
The arguments in this section, except for Lemma \ref{LemmaNonNullTF}
and Proposition
\ref{PropConditionalEqualityPXi}, closely mirrors the arguments given
in \cite{Mars99}, with several technical changes to allow the
application to electrovac space-times. Using the precise statement of
Theorem \ref{ThmMain}, $C_3$ should be taken to be 0 in this section.
We keep the notation $C_3$ to make explicit the applicability of the
computations in the global case.

We start first with some consequences of assumption (L)
\begin{lemma}\label{LemmaNonNullTF}
If $\mathcal{B}_{ab}$ vanishes identically on $\mathcal{M}$, then we
have that
\begin{enumerate}
\item $\mathcal{F}_{ab}\mathcal{F}^{ab}$ only vanishes on sets of
co-dimension $\geq 1$,
\item $\mathcal{F}_{ab}\mathcal{F}^{ab} = 0 \implies \mathcal{F}_{ab}
= 0$,
\item The Killing vector field $t^a$ is non-null on a dense subset of
$\mathcal{M}$. 
\end{enumerate}
\end{lemma}
\begin{proof}
Squaring the alignment condition implied by the vanishing of
$\mathcal{B}_{ab}$ gives
\[ \mathcal{F}^2 = (4\bar{\Xi} - 2 \bar{C}_3)^2\mathcal{H}^2~. \]
By assumption (L), if the left-hand side vanishes, then $4\Xi -
2 C_3 = 0$, and using the alignment condition again, we have
$\mathcal{F}_{ab} = 0$. This proves claim (2). 

Suppose $\mathcal{F}_{ab}$ vanishes on some small open set $\delta$,
then necessarily $\nabla_at_b = 0$ on $\delta$. Furthermore, we have
that $\Xi$ must be locally constant as shown above, and thus
$\nabla_a\Xi = \mathcal{H}_{ba}t^b = 0$. But
\[ \nabla_a\Xi\nabla^a\Xi = \mathcal{H}_{ba}\mathcal{H}^{ca}t_ct^b =
\frac14\mathcal{H}_{ab}\mathcal{H}^{ab}t^ct_c = 0 \]
and since the Maxwell field is non-null, we have that $t^a$ must be a
parallel null vector in $\delta$. If $t^a$ is not the zero vector,
however, we must have $t^a$ being an eigenvector, and hence a
principal null direction, of $\mathcal{H}_{ab}$, with eigenvalue zero:
this contradicts the fact that $\mathcal{H}_{ab}$ is non-null. If
$t^a= 0$ on a small neighborhood $\delta$, however, $t^a$ must vanish
everywhere on $\mathcal{M}$ since it is Killing, contradicting
assumption (A4). This proves assertion (1). 

Lastly, assume that $t^2 = 0$ on some small open set $\delta$, which
implies $\nabla_at^2 = 0$ and $\Box_g t^2 = 0$ on the neighborhood. Using
(\ref{EqFStarFStar}), we deduce 
\[ {}^*F_{mx}t^x{}^*F_{ny}t^y = \frac12 F^2 t_mt_n \]
Taking the trace in $m,n$, we have
\[ {}^*F_{mx}{}^*F^{my}t_yt^x = 0 \]
Using the fact that
\[ F_{mx}F^{my}t^xt_y = \nabla_mt^2\nabla^mt^2 = 0~,\qquad \mathcal{F}_{ac}\bar{\mathcal{F}}_b{}^c = \frac14(F_{ac}F_b{}^c +
{}^*F_{ac}{}^*F_b{}^c) \]
we have
\[ \mathcal{F}_{ac}\bar{\mathcal{F}}_b{}^ct^at^b = 0 \]
Now, since $\mathcal{B}_{ab} = 0$, this implies that
\[ |2 C_3 - 4\Xi|^2 T_{ab}t^at^b = 0 \]
on the open set $\delta$. If the first factor is identically zero in
an open subset $\delta'\subset\delta$, then $\Xi$ is locally constant
and arguing the same way as above we get a contradiction. Therefore
we can assume, without loss of generality, that $T_{ab}t^at^b = 0$.
Now consider the identity
\[ 0 = \Box_gt^2 = \nabla^b(t^aF_{ba}) = \frac12 F^{ba}F_{ba} - 2
R_{ab}t^at^b \]
The last term vanishes by the assumption, and implies that
$F^{ba}F_{ba} = 0$; thus ${}^*F_{mx}t^x = 0$. Therefore
\[ \nabla_at^2 = t^bF_{ab} = 2 t^b\mathcal{F}_{ab} \]
in $\delta$, and hence
\[ 0 = \Box_g t^2 = \mathcal{F}_{ab}\mathcal{F}^{ab} - 2 R_{ab}t^at^b
\]
and so $\mathcal{F}_{ab}\mathcal{F}^{ab} = 0$ identically on $\delta$,
which we have just shown is impossible. Assertion (3) then follows.
\end{proof}
We can then prove claim (1) in Theorem \ref{ThmMain}:
\begin{prop}\label{PropConditionalEqualityPXi}
If $\mathcal{B}_{ab}$ and $\mathcal{Q}_{abcd}$ both vanish on
$\mathcal{M}$, then $P^{-1} - 2\Xi$ is constant.
\end{prop}
\begin{proof}
We start by calculating $\mathcal{H}^{ab}\nabla_c\mathcal{B}_{ab}$.
Using (\ref{EqDerASDErnst}), 
\begin{align*}
\mathcal{H}^{ab}\nabla_c\mathcal{F}_{ab} & = 2[\mathcal{Q}_{dcab} +
3P(\mathcal{F}\tilde\otimes\mathcal{H})_{dcab}]t^d\mathcal{H}^{ab}
\\
& \qquad + \frac12(T_{ad}\mathcal{H}^a{}_c + T_{bc}\mathcal{H}_d{}^b -
T_{ac}\mathcal{H}^a{}_d - T_{bd}\mathcal{H}_c{}^b)t^d\\
& \qquad + i(T_d{}^e{}^*\mathcal{H}_{ec} +
T_c{}^f{}^*\mathcal{H}_{df})t^d\\
&= 2[ \mathcal{Q}_{dcab} +
3P(\mathcal{F}\tilde\otimes\mathcal{H})_{dcab}]t^d\mathcal{H}^{ab} +
2(T_{ad}\mathcal{H}^a{}_c + T_{bc}\mathcal{H}_d{}^b)t^d \\
& = 2\mathcal{Q}_{dcab}\mathcal{H}^{ab}t^d +
P(3\mathcal{F}_{dc}\mathcal{H}_{ab}\mathcal{H}^{ab} +
\mathcal{H}_{dc}\mathcal{F}_{ab}\mathcal{H}^{ab})t^d \\
& \qquad + 8(\mathcal{H}_{af}\bar{\mathcal{H}}_d{}^f\mathcal{H}^a{}_c
+ \mathcal{H}_{bf}\bar{\mathcal{H}}_c{}^f\mathcal{H}_d{}^b)t^d \\
& = 2\mathcal{Q}_{dcab}\mathcal{H}^{ab}t^d +
P(3[\mathcal{B}_{dc}- (2\bar{C}_3-4\bar{\Xi})\mathcal{H}_{dc}]\mathcal{H}_{ab}\mathcal{H}^{ab} \\
& \qquad + \mathcal{H}_{dc}[\mathcal{B}_{ab} - (2\bar{C}_3 -
4\bar{\Xi})\mathcal{H}_{ab}]\mathcal{H}^{ab})t^d 
+ 4\mathcal{H}_{ab}\mathcal{H}^{ab}\bar{\mathcal{H}}_{dc}t^d
\end{align*}
where we used (\ref{EqSelfDualSelfProduct}) and (\ref{EqDefB}) in the
last equality. Using (\ref{EqDefP}), we simplify to 
\begin{align*} \mathcal{H}^{ab}\nabla_c\mathcal{F}_{ab} &=
2\mathcal{Q}_{dcab}\mathcal{H}^{ab}t^d -
\frac{3}{C_1^2P^3}\mathcal{B}_{dc}t^d +
\frac{4}{C_1^2P^3}(2\bar{C}_3-4\bar{\Xi})\mathcal{H}_{dc}t^d \\
& \qquad + \mathcal{H}_{dc}\mathcal{B}_{ab}\mathcal{H}^{ab}t^d -
\frac{4}{C_1^2P^4}\bar{\mathcal{H}}_{dc}t^d\end{align*}
Applying the condition $\mathcal{Q}_{abcd} = 0$ and $\mathcal{B}_{ab}
= 0$ and (\ref{EqDefXi}), we have
\[ \mathcal{H}^{ab}\nabla_c\mathcal{F}_{ab} =
\frac{4}{C_1^2P^3}(2\bar{C}_3 - 4 \bar{\Xi}) \nabla_c\Xi -
\frac{4}{C_1^2P^4}\nabla_c\bar{\Xi} \]
On the other hand, we can calculate
\[
\mathcal{H}^{ab}\nabla_c\left[(2\bar{C}_3-4\bar{\Xi})\mathcal{H}_{ab}\right] =
-4 \mathcal{H}^{ab}\mathcal{H}_{ab}\nabla_c\bar{\Xi} +
\frac12(2\bar{C}_3-4\bar{\Xi})\nabla_c(\mathcal{H}_{ab}\mathcal{H}^{ab})
\]
So putting them altogether we have
\[ 0 = \mathcal{H}^{ab}\nabla_c\mathcal{B}_{ab} =
\frac{4}{C_1^2P^3}(\bar{C}_3 - 2 \bar{\Xi}) (2\nabla_c\Xi -
\nabla_c\frac{1}{P}) \]
By the arguments used in the proof of Lemma \ref{LemmaNonNullTF},
$\Xi$ is not locally constant and so 
$C_3 \neq 2\Xi$ densely. The above expression 
(and continuity) then shows that $2\Xi - \frac{1}{P}$ is constant.  
\end{proof}

In what follows I'll write $C_2 = P^{-1} - 2\Xi + C_3$.
\begin{rmk}
In the global case (where we assume (G) instead of (L)), the decay
condition given by asymptotic flatness shows that $2\Xi$ and $1/P$ 
both vanish at spatial infinity, and so $C_2 = C_3 = M/(q_E - i q_B)$ everywhere. 
\end{rmk}

The next proposition demonstrates assertion (2) in Theorem
\ref{ThmMain}.
\begin{prop}\label{PropPIdentities}
Assuming the vanishing of $\mathcal{B}_{ab}$ and $\mathcal{Q}_{abcd}$,
we have the following identities
\begin{subequations}
\begin{align}
t^2 &= -\left| \frac{1}{P} - C_2 \right|^2 + C_4\\
(\nabla P)^2 & = -\frac{t^2}{C^2_1}\\
C_1\Box_g P &= - \frac{2}{C_1\bar{C}_1P\bar{P}}\left( \bar{C}_1C_2 -
(|C_2|^2 - C_4)\bar{C}_1\bar{P}\right)
\end{align}
\end{subequations}
where $C_4$ is a real-valued constant.
\end{prop}
\begin{proof}
We can calculate
\[ \nabla_a t^2 = 2t^b\nabla_at_b = -F_{ba}t^b =
-2\Re[\mathcal{F}_{ba}t^b] \]
The vanishing of $\mathcal{B}_{ab}$ and Proposition
\ref{PropConditionalEqualityPXi} together imply
\[ \nabla_a t^2 = -4 \Re[ (2\bar\Xi - \bar{C}_3)\mathcal{H}_{ba}t^b]
= -2 \Re[ (\frac{1}{\bar P} - \bar{C}_2)\nabla_a\frac{1}{P}] = -\nabla_a
\left| \frac{1}{P} - C_2 \right|^2
\]
The first claim follows as $\mathcal{M}$ is simply connected. 
Next, from Proposition \ref{PropConditionalEqualityPXi} we get
\[ \nabla_a P = \nabla_a \frac{1}{2\Xi + C_2 - C_3} = -
\frac{2\nabla_a\Xi}{(2\Xi+ C_2 - C_3)^2} = -2P^2 \mathcal{H}_{ba}t^b \]
So
\[ \nabla_aP\nabla^aP = 4P^4 \mathcal{H}_{ba}t^b\mathcal{H}^{ca}t_c =
P^4 \mathcal{H}^2 t^2 = - \frac{t^2}{C_1^2} \]
where we used (\ref{EqSelfDualSelfProduct}) and the definition for
$P$. We can also calculate directly the D'Alembertian
\begin{align*}
\Box_g P & = - 2\nabla^a(P^2\mathcal{H}_{ba}t^b) \\
& = -2 \mathcal{H}_{ba} (2P\nabla^aP t^b + \frac12 P^2F^{ab}) \\
& = 2 \mathcal{H}_{ba}(4P^3\mathcal{H}^{ca}t_ct^b + \frac12
P^2\mathcal{F}^{ba}) \\
& = 2 P^3 \mathcal{H}^2 t^2 + 2 P^2 (\frac{1}{\bar{P}}-\bar{C}_2)\mathcal{H}^2 \\
& = 2P^2\mathcal{H}^2 \left[ P\left(-\left| \frac{1}{P} - C_2
\right|^2 + C_4\right) + \frac{1}{\bar{P}}-\bar{C}_2\right] \\
& = 2P^2\mathcal{H}^2 \left[ \left(\frac{1}{\bar{P}} -
\bar{C}_2\right)\left(1-P\left(\frac{1}{P}-C_2\right)\right) + C_4 P \right] \\
&= \frac{2}{C_1^2 P}\left( C_2(\frac{1}{\bar{P}} - \bar{C}_2) + C_4\right)
\end{align*}
from which the third identity follows by simple algebraic
manipulations.
\end{proof}

\begin{rmk}\label{RmkPIdentities} If we further impose the condition that $C_1\bar{C}_2$ is
real, then the imaginary part of the third identity becomes
\[ \Im ( \Box_g C_1 P) = \frac{ 2( | C_2|^2 - C_4) }{|C_1 P|^2} \Im
(C_1 P)  \]
which will be useful later.
In the global case, we can again match the data at spatial infinity to
see that $C_4 = |C_2|^2 - 1 = M^2/q^2 - 1$ (the condition relating
$C_2$ and $C_4$ in Theorem \ref{ThmMain} is directly satisfied); the 
third identity then 
reads:
\[ (q_E + i q_B)\Box_gP = - \frac{2}{q^2P\bar{P}}\left( M - (q_E - i q_B)
\bar{P}\right) \]
\end{rmk}
An immediate consequence of the above proposition is that $(\nabla C_1
P)^2$ is real. Writing the complex quantity $C_1 P=y+iz$, where $y$
and $z$ are real-valued, we see that this
implies 
\[\nabla^ay \nabla_az = 0\]
Furthermore, by Lemma \ref{LemmaNonNullTF}, we have that, with the possible
exception on sets of co-dimension $\geq 1$, $t^2\neq 0$. This leads to the 
useful
observation that, with the possible exception on those points, 
$(\nabla y)^2$ and $(\nabla z)^2$ cannot simultaneously vanish, and 
in particular $\nabla_a y$ and $\nabla_a z$ are not simultaneously 
null, and thus rule out the case where the two are aligned. We
summarize in the following
\begin{cor}\label{CorPropertiesyz}
Letting $C_1P=y+iz$, we know that on any open set
\begin{enumerate}
\item $P$ is not locally constant
\item $\nabla_a y$ and $\nabla_a z$ are mutually orthogonal
\item $\nabla_a y$ and $\nabla_a z$ cannot be both null
\item $\nabla_a y$ and $\nabla_a z$ cannot be parallel
\end{enumerate}
\end{cor}
Replacing $C_1P$ by $y+iz$, and imposing the condition $C_1\bar{C}_2$
is real, we can also rewrite
\[  t^2 = - \frac{C_1\bar{C}_1 - 2C_1\bar{C}_2 y}{y^2 +
z^2}- |C_2|^2  + C_4 \]

Since $\mathcal{H}_{ab}$ is an anti-self-dual two form with 
non-vanishing norm, it has two distinct principal null directions,
which we denote by $\lbar^a$ and $l^a$, with the normalization 
$g_{ab}\lbar^al^b = -1$. The alignment of $\mathcal{H}_{ab}$ with
$\mathcal{F}_{ab}$ (via vanishing of $\mathcal{B}_{ab}$) allows the
following expressions
\begin{align*}
\mathcal{H}_{ab} &= \frac{1}{2C_1P^2}(\lbar_al_b - l_a\lbar_b +
i\epsilon_{abcd}\lbar^cl^d)\\
\mathcal{F}_{ab} &=
\frac{\frac{1}{\bar{P}}-\bar{C}_2}{C_1P^2}(\lbar_al_b
- l_a\lbar_b + i\epsilon_{abcd}\lbar^cl^d)
\end{align*}

By the assumption $\mathcal{Q}_{abcd}=0$, the principal null
directions of $\mathcal{H}_{ab}$ are repeated null directions of the
anti-self-dual Weyl tensor, and thus the space-time is algebraically
special (Type D). On a local neighborhood, we can take $m,\bar{m}$ 
complex smooth vector fields to complete the null tetrad
$\{m,\bar{m},\lbar,l\}$ (see Appendix \ref{SectTetrad}), and in the
tetrad (spinor) formalism, the only non-zero Weyl scalar is 
\begin{subequations}\label{SubeqAlignment}
\begin{equation}
\Psi := \Psi_0 = W(\bar{m},\lbar,m,l) =
-\frac{1}{C_1^2P^3}\left(\frac{1}{\bar{P}} - \bar{C}_2\right) 
\end{equation}
the only non-zero component of the Maxwell scalars is
\begin{equation}
\Upsilon := \Upsilon_0 = \mathcal{H}_{ab}l^a\lbar^b = \frac{1}{2C_1P^2}
\end{equation}
and the only non-zero component of the Ricci scalars is
\begin{equation}
\Phi := \Phi_0 = T(\lbar, l) = T(m,\bar{m}) =
\frac{1}{C_1\bar{C}_1P^2\bar{P}^2}
\end{equation}
\end{subequations}
Notice the following symmetry relations
\begin{equation}\label{EqSymmetryScalarVars}
\bar\Psi = \ubar\Psi~,\quad \ubar{\bar\Psi} = \Psi~,\quad
\ubar{\bar\Upsilon} = -\Upsilon~,\quad \bar\Phi = \ubar\Phi = \Phi
\end{equation}

Now, from 
\[ 2C_1P^2\mathcal{H}_{ab}t^a = - C_1\nabla_b P \]
we can calculate
\begin{subequations}\label{SubeqDecompositionGradyz}
\begin{align}
\nabla_b y &= \lbar_b~l_at^a - l_b~\lbar_at^a & (\nabla
y)^2 &= 2l_a\lbar_bt^at^b \\
\nabla_b z &= \epsilon_{bacd}t^a\lbar^cl^d & (\nabla z)^2 &=
2\lbar_al_bt^at^b + t^2
\end{align}
\end{subequations}
So we need expressions for $g(t,\lbar),g(t,l)$. From the fact that
$\mathcal{L}_t\mathcal{H} = 0$, we have
\[ [t,\lbar]_al_b + \lbar_a[t,l]_b - [t,l]_a\lbar_b - l_a[t,\lbar]_b =
0 \]
which we can contract against $l$ and $\lbar$ (using the fact
that $[t,l]_al^a = \partial_t l^2 = 0$) to arrive at
\begin{subequations}
\begin{align}
[t,\lbar]_a &= \lbar_a[t,l]_b\lbar^b = K_t\lbar_a \\
[t, l]_a &= l_a[t,\lbar]_bl^b = -K_tl_a
\end{align}
\end{subequations}
where the function $K_t:= [t,l]_b\lbar^b$. Now
\[ \partial_t (t_b\lbar^b) = \mathcal{L}_t(t_b\lbar^b) = K_tt_b\lbar^b
\]
and similarly
\[ \partial_t (t_bl^b) = -K_tt_bl^b \]
Lastly, we compute an expression for $t$ by 
\[ -\frac{\mathcal{H}^{cb}\nabla_bP}{2P^2} =
\frac{1}{4}\mathcal{H}^2t^c = - \frac{t^c}{4C_1^2P^4} \]
Therefore, by a direct computation
\begin{equation}\label{EqDecompositiont} t_c = -(l_at^a)\lbar_c - (\lbar_at^a)l_c -
\epsilon_{cabd}(\nabla^az)\lbar^bl^d \end{equation}

Next is the main lemma of this section
\begin{lemma}\label{LemmaNormGradyz}
Assuming $\mathcal{B}_{ab}$ and $\mathcal{Q}_{abcd}$ vanish, 
$C_1\bar{C}_2$ is real, and $|C_2|^2 - C_4 = 1$, we have the norms
\begin{subequations}
\begin{align}
(\nabla z)^2 &= \frac{\mathfrak{A} - z^2}{y^2+z^2}\\
(\nabla y)^2 &= \frac{\mathfrak{A} + y^2 + |C_1|^2 -
2C_1\bar{C}_2y}{y^2+z^2}
\end{align}
\end{subequations}
where $\mathfrak{A}$ is a non-negative constant with $z^2 \leq
\mathfrak{A}$.
\end{lemma}
\begin{proof}
We will use the tetrad formalism of Klainerman-Ionescu (see Appendix
\ref{SectTetrad}) extensively in the following computation. By the
alignment properties (\ref{SubeqAlignment}) and the symmetry
properties (\ref{EqSymmetryScalarVars}), the Maxwell equations
simplify to 
\begin{align*}
\ubar{D}\Upsilon &= -2\ubar{\bar\theta}\Upsilon & D\Upsilon &=
-2\theta\Upsilon \\
-\delta\Upsilon &= 2\eta\Upsilon & -\bar{\delta}\Upsilon &= 2
\ubar{\bar\eta}\Upsilon
\end{align*}
from which we arrive at 
\begin{equation}\label{EqNullTetradGradP}
DP = \theta P~,\quad \ubar{D}P = \ubar{\bar\theta}P~, \quad \delta P =
\eta P~,\quad \bar\delta P = \ubar{\bar\eta}P
\end{equation}
From the decomposition (\ref{SubeqDecompositionGradyz}) we then have
\begin{subequations}
\begin{align}
\label{EqGradyTetradCoeff}\nabla_a y &= - \theta C_1P \lbar_a -
\ubar{\bar\theta}C_1P l_a\\
\label{EqGradzTetradCoeff}i\nabla_a z &= \eta C_1P \bar{m}_a +
\ubar{\bar\eta}C_1Pm_a
\end{align}
\end{subequations}
Using the fact that $y$ and $z$ are real, taking complex conjugates on
the above equations gives us
\begin{equation}\label{EqRealityPTimesTetradCoeff}
\theta C_1P = \bar\theta\bar{C}_1\bar P~,\quad \ubar\theta\bar{C}_1\bar P =
\ubar{\bar\theta}C_1P~,\quad \eta C_1P = -\ubar\eta\bar{C}_1\bar P
\end{equation}
The Bianchi equations become
\begin{subequations}
\begin{align}
\label{EqBianchiRel1} 0 &= \xi (3\Psi + \Phi) \\
\label{EqBianchiRel2} 0 &= \vartheta (3\Psi - \Phi) \\
-D(\Psi + \frac12\Phi) &= 3\theta\Psi + \bar\theta\Phi \\
\bar\delta(\Psi - \frac12\Phi) &= -3\ubar{\bar\eta}\Psi +
\bar\eta\Phi\\
-\delta\Phi &= 2(\eta+\ubar\eta)\Phi \\
D\Phi &= -2(\bar\theta +\theta)\Phi
\end{align}
\end{subequations}
Because of the triple alignment given by $\mathcal{B}_{ab} = 0$ and
$\mathcal{Q}_{abcd} = 0$, the latter four equations contain
essentially the same information as the Maxwell equations. We examine
the first two in more detail. Consider the equation $3\Psi \pm \Phi =
0$. This implies
\[ 3\bar{C}_1\bar{C}_2\bar{P}^2 - 3 \bar{C}_1\bar{P} \pm C_1 P = 0 \]
or
\begin{align*}
\frac{3 C_1 \bar{C}_2}{C_1\bar{C}_1} (y^2 - z^2)  -(3\mp 1) y  &= 0\\
\frac{6 C_1\bar{C}_2}{C_1\bar{C}_1} yz - (3\pm 1) z &= 0
\end{align*}
Taking derivatives, we have
\begin{align*}
(\frac{6 C_1 \bar{C}_2}{C_1\bar{C}_1}y - 3 \pm 1)\nabla y &= \frac{6 C_1
\bar{C}_2}{C_1\bar{C}_1}z\nabla z \\
(\frac{6 C_1\bar{C}_2}{C_1\bar{C}_1} y - 3 \mp 1)\nabla z &= - \frac{6
C_1\bar{C}_2}{C_1\bar{C}_1} z\nabla y
\end{align*}
By assumption that $C_1\bar{C}_2$ is real, all the coefficients in the
above two equations are real. Suppose the equation $3\Psi\pm\Phi = 0$ is satisfied on an open-set, 
as $\nabla y$ and $\nabla z$ cannot be parallel by Corollary
\ref{CorPropertiesyz}, we must have then 
\[ (\frac{6 C_1 \bar{C}_2}{C_1\bar{C}_1}y - 3 \pm 1)\nabla y =
\frac{6 C_1\bar{C}_2}{C_1\bar{C}_1}z\nabla z= 0 \]
This implies that $y$ and $z$ are locally constant, which contradicts
statement (1) in Corollary \ref{CorPropertiesyz}. Therefore an equation of the form $3\Psi\pm \Phi = 0$
cannot be satisfied on open sets. 

Applying to the Bianchi identities
(\ref{EqBianchiRel1},\ref{EqBianchiRel2}), we see that $\xi = \vartheta =
\ubar\xi = \ubar\vartheta = 0$. The relevant null structure equations,
simplified with the above observation, are
\begin{subequations}
\begin{align}
\label{EqNullStructRel1} (D+\Gamma_{124})\eta &=\theta(\ubar\eta -\eta)\\
\label{EqNullStructRel2} -\delta\theta &= \zeta\theta + \eta(\theta-\bar\theta)
\end{align}
\end{subequations}
Define the quantity $A = C_1\bar{C}_1P\bar{P}(\nabla z)^2$. Equations
(\ref{EqGradzTetradCoeff}) and (\ref{EqRealityPTimesTetradCoeff}) 
implies that $(\nabla z)^2 = 2 \eta\bar{\eta}C_1\bar{C}_1P\bar{P}$, so
\begin{align*}
0 \leq A &= 2\eta\bar\eta C_1^2\bar{C}_1^2P^2\bar{P}^2\\
&= 2C_1^2\bar{C}_1^2\ubar\eta\ubar{\bar\eta} P^2\bar{P}^2 \\
& = - (y^2+z^2) - (C_1\bar{C}_1 - 2 C_1\bar{C}_2 y) -
2\theta\ubar{\theta}C_1^2\bar{C}_1^2 P^2\bar{P}^2
\end{align*}
where in the last line we used Proposition \ref{PropPIdentities},
Corollary \ref{CorPropertiesyz}, Equations
(\ref{EqGradyTetradCoeff}) and (\ref{EqRealityPTimesTetradCoeff}), and
the assumption that $|C_2|^2 - C_4 = 1$.
By using (\ref{EqNullStructRel1},\ref{EqNullStructRel2}) we calculate
\begin{align*}
D(\eta\bar\eta) &= \theta(\ubar\eta-\eta)\bar\eta +
\bar\theta(\ubar{\bar\eta}-\bar\eta)\eta \\
\delta(\theta\ubar\theta) &= -\eta(\theta-\bar\theta)\ubar\theta -
\ubar\eta(\ubar\theta -\ubar{\bar\theta})\theta
\end{align*}
Thus, with judicial applications of (\ref{EqRealityPTimesTetradCoeff})
\begin{align*}
DA &= 2C_1^2\bar{C}_1^2[ \theta(\ubar\eta-\eta)\bar\eta +
\bar\theta(\ubar{\bar\eta}-\bar\eta)\eta ] P^2\bar{P}^2 +
4C_1^2\bar{C}_1^2\eta\bar\eta (\theta + \bar\theta) P^2\bar{P}^2 \\
& = 0\\
\delta A &= -\delta(z^2) +
2C_1^2\bar{C}_1^2P^2\bar{P}^2[\eta(\theta-\bar\theta)\ubar\theta +
\ubar\eta(\ubar\theta -\ubar{\bar\theta})\theta] \\
& \qquad -
4C_1^2\bar{C}_1^2P^2\bar{P}^2(\eta + \ubar{\eta})\theta\ubar{\theta}
\\ &= -\delta(z^2)
\end{align*}
Since $Dz = \ubar{D}z = 0$, we have that the function $A+z^2$ is constant. Define $\mathfrak{A} = A+z^2$. 
The nonnegativity of $A$
guarantees that $z^2 \leq \mathfrak{A}$, and we have
\[ (\nabla z)^2 = \frac{A}{C_1\bar{C}_1P\bar{P}} =
\frac{\mathfrak{A}-z^2}{y^2+z^2} \]
and
\[ (\nabla y)^2 = (C_1\nabla P)^2 + (\nabla z)^2 =
\frac{\mathfrak{A}+y^2 + C_1\bar{C}_1 - 2 C_1\bar{C}_2y}{(y^2+z^2)} \]
as claimed.
\end{proof}

\begin{rmk}
In the proof above we showed that $\xi = \vartheta = \ubar\xi
=\ubar\vartheta = 0$, a conclusion that in the vacuum case
\cite{Mars99} is easily reached by the Goldberg-Sachs theorem. It is
worth noting that in general, the alignment of the principal null
directions of the Maxwell form and the Weyl tensor is not enough to
justify the vanishing of all four of the involved quantities. Indeed,
the Kundt-Thompson theorem \cite{ExactSolutions} only guarantees that
$\xi\vartheta = \ubar\xi\ubar\vartheta = 0$. In our special
case the improvement comes from the fact that we not only have
alignment of the principal null directions, but also knowledge
of the proportionality factor. This allows us to write down the
polynomial expression in $P$ and $\bar{P}$ which we used to eliminate
the case where only one of $\xi$ and $\vartheta$ vanishes. 
\end{rmk}

In the remainder of this section, we assume that $C_1\bar{C}_2$ is
real and $|C_2|^2 - C_4 = 1$ and prove assertions (3) and (4) in
Theorem \ref{ThmMain}. Let us first define two auxillary vector fields. 
On our space-time, let
\begin{equation}\label{EqOriginalDefn}
n^a = (\mathfrak{A}+y^2)t^a + (y^2+z^2)(t_bl^b\lbar^a
+ t_b\lbar^bl^a)
\end{equation}
Define $\mathcal{M}_\mathfrak{A}:= \{ p\in\mathcal{M} |
z^2(p)<\mathfrak{A} \}$. On this open subset we can define 
\begin{equation}
b^a = \frac{\nabla^a z}{(\nabla z)^2} =
\frac{y^2+z^2}{\mathfrak{A} - z^2}\nabla^a z
\end{equation}

We also define the open subsets $\mathcal{M}_l := \{ p\in\mathcal{M} |
(t_al^a)(p) \neq 0 \}$ and $\mathcal{M}_\lbar := \{ p\in\mathcal{M} |
(t_a\lbar^a)(p) \neq 0 \}$. Now, notice that in our calulcations above
using the tetrad formalism, we have only specified the ``direction''
of $l,\lbar$ and their lengths relative to each other. We still have
considerable freedom left to fix the lapse of one of the two vector
fields and still retain the use of our formalism. On
$\mathcal{M}_\lbar$, we can choose the vector field $\lbar$ such that 
$t_a\lbar^a = 1$ (similarly
for $l$ on $\mathcal{M}_l$; the calculations with respect to
$\mathcal{M}_l$ are almost identical to that on $\mathcal{M}_\lbar$,
so without loss of generality, we will perform calculations below with
respect to $\mathcal{M}_{\lbar}$) and the vector field $l$ maintaining
$l_a\lbar^a = -1$. From (\ref{SubeqDecompositionGradyz}) and  Lemma
\ref{LemmaNormGradyz}, we have that on $\mathcal{M}_\lbar$ we can
write
\begin{equation}\label{EqMlbarDecompositionGrady}
\nabla_ay = -l_a +
\frac{\mathfrak{A}+y^2+ |C_1|^2 -2 C_1\bar{C}_2y}{2(y^2+z^2)}\lbar_a = -l_a
+ U\lbar_a
\end{equation}
which implies $l_at^a = U$, where $U$ is defined on the entirety of
$\mathcal{M}$ as 
\begin{equation}
U := l_at^a\lbar_bt^b = \frac12(\nabla y)^2
\end{equation}

We consider first a special case when $t^a$ is hypersurface
orthogonal. 
\begin{prop}\label{PropCharacterHypersurfaceOrthogonal}
The following are equivalent:
\begin{enumerate}
\item $z$ is locally constant on an open subset
$\mathcal{U}\subset\mathcal{M}$
\item $\mathfrak{A}$ vanishes on $\mathcal{M}$
\item $z$ vanishes on $\mathcal{M}$
\end{enumerate}
\end{prop}
\begin{proof}
$(2)\implies (3)$ and $(3)\implies (1)$ follows trivially from Lemma
\ref{LemmaNormGradyz}. It thus suffices to show $(1)\implies(2)$.
Suppose $\nabla z = 0 |_\mathcal{U}$. We consider the imaginary part
of the third identity in Proposition \ref{PropPIdentities} \`a la
Remark \ref{RmkPIdentities}, which
shows that $z=0|_\mathcal{U}$. From Lemma \ref{LemmaNormGradyz} we
have $\mathfrak{A} = 0 |_\mathcal{U}$, but $\mathfrak{A}$ is a
universal constant for the manifold, and thus vanishes identically.
\end{proof}
It is simple to check that $z=0$ on $\mathcal{M}$ implies 
$C_1^{-1}\mathcal{H}_{ab}t^a =
\nabla_b\frac{1}{C_1P}$ is real, and so the vanishing of
$\mathcal{B}_{ab}$ implies $\mathcal{F}_{ab}t^a =
2(\frac{C_1\bar{C}_1}{\bar{C}_1\bar{P}} -
C_1\bar{C}_2)C_1^{-1}\mathcal{H}_{ab}t^a$ is purely real, which
by Frobenius' theorem gives that $t^a$ is hypersurface
orthogonal.\footnote{As to the question whether $t^a$ can be
hypersurface orthogonal without $\nabla z = 0$: in the next part we
will consider the case where $\mathfrak{A}\neq 0$ (implying $z$ is
nowhere locally constant), and show that in the subset
$\mathcal{M}_\lbar\cap\mathcal{M}_\mathfrak{A}$ we have local
diffeomorphisms to the Kerr-Newman space-time with non-zero angular
momentum, which implies that $\mathfrak{A} = 0$ is characteristic of
the Reissner-Nordstr\"om metric. Indeed, as we shall see later, the
quantity $\mathfrak{A}$ is actually square of the normalized angular
momentum of the space-time.} 
\begin{prop}\label{PropLocalIsometryRN}
Assume $\mathfrak{A} = 0$. Then, at any point $p\in \mathcal{M}_\lbar$
there exists a neighborhood that can be isometrically embedded into
the Reissner-Nordstr\"om solution.
\end{prop}
This proof closely mirrors that of Proposition 2 in \cite{Mars99}.
\begin{proof}
We use the same tetrad notation as before. Since $z=0$, we have $C_1P =
y$ is real, and hence (\ref{EqRealityPTimesTetradCoeff}) implies that
$\theta,\ubar\theta$ are real. Furthermore, $z=0$ implies via
(\ref{EqGradzTetradCoeff}) that $\eta=0=\ubar\eta$. The commutator
relations then gives
\begin{align*}
[D,\ubar{D}] &= -\ubar\omega D + \omega\ubar{D} \\
[\delta,\bar\delta] &= \Gamma_{121}\bar\delta + \Gamma_{122}\delta
\end{align*}
which implies that $\{l,\lbar\}$ and $\{m,\bar{m}\}$ are integrable.
Thus a sufficiently small neighborhood $\mathcal{U}$ can be
foliated by 2 mutually orthogonal families of surfaces. We calculate
the induced metric on the surface tangent to $\{m,\bar{m}\}$ using the
Gauss equation. 

First we calculate the second fundamental form $\chi(X,Y)$ for $X^a =
X_1m^a + X_2\bar{m}^a$ and $Y^a = Y_1m^a + Y_2\bar{m}^a$. By
definition $\chi(X,Y)$ is the projection of $\nabla_XY$ to the normal
bundle, so in the tetrad frame, evaluating using the connection
coefficients, we have
\begin{align*}
\chi(X,Y)^a &= X_1Y_1(\Gamma_{131}l^a + \Gamma_{141}\lbar^a) +
X_1Y_2(\Gamma_{231}l^a + \Gamma_{241}\lbar^a) \\
& \qquad + X_2Y_1(\Gamma_{132}l^a
+ \Gamma_{142}\lbar^a) + X_2Y_2(\Gamma_{232}l^a +
\Gamma_{242}\lbar^a)\\
& = X_1Y_1(\ubar\vartheta l^a + \vartheta \lbar^a) +
X_1Y_2(\ubar{\bar\theta}l^a + \bar\theta\lbar^a) \\
& \qquad + X_2Y_1(\ubar\theta l^a + \theta\lbar^a) +
X_2Y_2(\ubar{\bar\vartheta}l^a + \bar\vartheta\lbar^a)\\
&= - \frac{\nabla^a y}{C_1P}g(X,Y) = - \frac{\nabla^a y}{y}g(X,Y)
\end{align*}
where the last line used the vanishing of $\vartheta$ derived in the
proof of Lemma \ref{LemmaNormGradyz} and Equation
(\ref{EqGradyTetradCoeff}). We recall the Gauss equation
\[ R_0(X,Y,Z,W) = R(X,Y,Z,W) - g(\chi(X,W),\chi(Y,Z)) +
g(\chi(X,Z),\chi(Y,W)) \]
where $X,Y,Z,W$ are spanned by $m,\bar{m}$. Plugging in the explicit
form of the Riemann curvature tensor, we can compute by taking $X = Z
= m, Y = W = \bar{m}$ the only component of the curvature tensor for
the 2-surface
\begin{align*}
R_0(m,\bar{m},m,\bar{m})  & = -\Psi - \bar{\Psi} - \Phi - \frac{(\nabla
y)^2}{y^2} \\
& = \frac{C_1\bar{C}_1}{y^4} - \frac{2C_1\bar{C}_2}{y^3} -
\frac{(\nabla y)^2}{y^2} = - \frac{1}{y^2}\end{align*}
using Lemma \ref{LemmaNormGradyz} in the last equality. Now, since $\delta y = 0$, we have that the scalar curvature is constant on
the 2-surface, and positive, which means that its induced metric is 
locally the standard metric for $S^2$ with radius $|y|$. Now, since
$\nabla y\neq 0$ on our open set, it is possible to choose a local
co\"ordinate system $\{x^0, y,x^2,x^3\}$ compatible with the
foliation. Looking at (\ref{EqDecompositiont}) we see that $t^a$ is
non-vanishing inside $\mathcal{M}_\lbar$, and is in fact tangent to
the 2-surface formed by $\{l,\lbar\}$, so we can take $t = t_x
\partial_{x^0}$ for some function $t_x$. The fact that $t^a$ is
Killing gives that $\partial_{x^A}t_x = 0$ for $A=2,3$. Recall that
we are working in $\mathcal{M}_\lbar$, and we assumed that $\lbar_at^a
= 1$, then we can write, by (\ref{EqMlbarDecompositionGrady}), $\lbar
= \partial_y + s_x\partial_{x^0}$ for some function $s_x$. The
commutator identity
\[ [D,\delta] = -(\Gamma_{124}+\bar\theta)\delta + \zeta D \]
shows that $\partial_{x^A}s_x = 0$ by considering the decomposition
we have for $\lbar$ in terms of the co\"ordinate vector fields.
Then the Killing relation $[t,\lbar]=0$, together with the above,
implies that we can chose a co\"ordinate system $\{u,y,x^2,x^3\}$ with
$\partial_u = t$ and $\partial_y = \lbar$ that is compatible with the
foliation. Lastly, we want to calculate
$g_{AB} = g(\partial_{x^A},\partial_{x^B})$ in this
co\"ordinate system. To do so, we use the fact that 
\[ -l^a = t^a + U\lbar^a \]
Then the second fundamental form can be written as
\begin{align*}
\chi(X,Y) &= (\nabla_X Y)^\perp\\
&= -(\nabla_X Y)^a(\lbar_al^b + l_a\lbar^b) \\
&= -(\nabla_X Y)^a(\lbar_al^b -
(t_a+U\lbar_a)\lbar^b)
\end{align*}
Now, when $X,Y$ are tangential fields, since $U$ only depends on $y$
(recall that $z=\mathfrak{A}=0$),
we have that $\nabla_XU = 0$. Furthermore, we use $g(Y,l) = g(Y,t) =
0$ to see 
\[ \chi(X,Y)^b = l^b Y^aX^c\nabla_c\lbar_a - \lbar^bY^aX^c\nabla_ct_a
- \lbar^bUY^aX^c\nabla_c\lbar_a \]
So we have, using the fact that the second fundamental form is symmetric
\begin{align*}
2\chi(X,Y)^b &= Y^aX^c(l^b - U\lbar^b)\mathcal{L}_\lbar g_{ac} -
Y^aX^c\lbar^b\mathcal{L}_tg_{ac} \\
& = - Y^aX^c\mathcal{L}_{\lbar}g_{ac} \nabla^b y
\end{align*}
Taking $X$ and $Y$ to be co\"ordinate vector fields, we conclude that
\[ \partial_y g_{AB} = \frac{2}{y}g_{AB} \]
so that $g_{AB} = y^2 g^0_{AB}$ where $g^0_{AB}$ only depends on
$x^2,x^3$. Imposing the condition that $g_{AB}$ be the matrix for
the standard metric on a sphere of radius $|y|$, 
we finally conclude
that the line element can be written as
\[ ds^2 = -(1-\frac{2C_1\bar{C}_2y-|C_1|^2}{y^2})du^2 + 2dudy +
y^2d\omega_{\mathbb{S}^2} \]
and thus the neighborhood can be embedded into Reissner-Nordstr\"om
space-time of mass $C_1\bar{C}_2$ and charge $|C_1|$.
\end{proof}

Notice that a priori there is no guarantee that $C_1\bar{C}_2y > 0$,
this is compatible with the fact that we did not specify, for the
local version of the theorem, the requirement for asymptotic flatness,
and hence are in a case where the mass is not necessarily positive.

Next we consider the general case where $t^a$ is not hypersurface
orthogonal. In view of Proposition
\ref{PropCharacterHypersurfaceOrthogonal}, we can assume that
$\mathfrak{A} > 0$ and $z$ not locally constant on any open set. Then
it is clear that the set $\mathcal{M}_\mathfrak{A}$ is in fact dense
in $\mathcal{M}$: for if there exists an open set on which $z =
\mathfrak{A}$, then Proposition
\ref{PropCharacterHypersurfaceOrthogonal} implies that $\mathfrak{A} =
0$ identically on $\mathcal{M}$. Therefore, the set
$(\mathcal{M}_\lbar\cup\mathcal{M}_l)\cap\mathcal{M}_\mathfrak{A}$ is
non-empty as long as $\mathcal{M}_\lbar\cup\mathcal{M}_l$ is
non-empty; this latter fact can be assured since by assumption (A4)
that $t^a$ is timelike at some point $p\in\mathcal{M}$, whereas $l^a$ and $\lbar^a$ are
non-co\"incidental null vectors, so in a neighborhood of $p$, 
we must have $l^at_a\neq 0 \neq \lbar^at_a$. It is on this
set that we consider the next proposition. 

\begin{prop}\label{PropLocalIsometryKerrNewman}
Assuming $\mathfrak{A} > 0$. Let
$p\in\mathcal{U}\subset\mathcal{M}_{\lbar}\cap\mathcal{M}_\mathfrak{A}$
such that $t^a,n^a,b^a$ and $\lbar^a$ are well-defined on
$\mathcal{U}$, with normalization $\lbar^at_a=1$. Then the four vector
fields form a holonomic basis, and $U$ can be isometrically embedded
into a Kerr-Newman space-time.
\end{prop}
Before giving the proof, we first record the metric for the
Kerr-Newman solution in Kerr co\"ordinates
\begin{align}
ds^2 &= -\left(1-\frac{2Mr-q^2}{r^2+a^2\cos^2\theta}\right)dV^2 + 2drdV
+ (r^2+a^2\cos^2\theta)d\theta^2\\
\nonumber &\qquad +
\frac{\left[(r^2+a^2)^2-(r^2-2Mr+a^2+q^2)a^2\sin^2\theta\right]\sin^2\theta}{r^2+a^2\cos^2\theta}d\phi^2\\
\nonumber &\qquad - 2a\sin^2\theta d\phi dr -
\frac{2a(2Mr-q^2)}{r^2+a^2\cos^2\theta}\sin^2\theta dVd\phi
\end{align}
Notice that the metric is regular at $r=M\pm\sqrt{M^2-a^2-q^2}$ the
event and Cauchy horizons.
\begin{proof}
We first note that in $\mathcal{M}_\lbar$, we have the normalization
\[ n^a = (y^2+z^2)(l^a + U\lbar^a) + (\mathfrak{A} + y^2)t^a
\]
For the proof, it suffices to establish that the commutators between
$n^a,b^a,\lbar^a,t^a$ vanish and that the vectors are linearly
independent (for holonomy), and to calculate their relative
inner-products to verify that they define a co\"ordinates equivalent
to the Kerr co\"ordinate above. 

First we show that the commutators vanish. The cases $[t,\cdot]$ are
trivial. Since we fixed $\lbar^at_a = 1$, we have that
\[ 0 = t^b\nabla_b(\lbar_at^a) = K_tt_b\lbar^b = K_t\]
so that $K_t = 0$ and thus $[t,\lbar] = [t,l] = 0$. Since $y$ and $z$
are geometric quantities defined from $\mathcal{H}_{ab}$, and $U$ is a 
function only of $y$ and $z$, they are symmetric under the action of
$t^a$, therefore $[t,n] = 0$. Similarly, to evaluate $[t,b]$, it
suffices to consider $[t,\nabla z]$. Using
(\ref{SubeqDecompositionGradyz})
we see that $\nabla z$ is defined by the volume form, the metric, and
the vectors $t^a,\lbar^a,l^a$, all of which symmetric under
$t$-action, and thus $[t,b]=0$. The remaining cases require consideration
of the connection coefficients. In view of the normalization condition
imposed, $\nabla_a y = -l_a + U\lbar_a$, so (\ref{EqGradyTetradCoeff})
implies $\ubar{\theta}\bar{C}_1\bar{P} = 1$, $\theta C_1P = -U$. Recall the null
structure equation
\[ -\delta\ubar{\theta} = -\zeta\ubar\theta + \ubar\eta(\ubar\theta
-\ubar{\bar\theta}) \]
Using 
\[ 0 = \delta(\ubar\theta \bar{C}_1\bar P) = (\delta\ubar\theta)
\bar{C}_1 \bar P +
\ubar\theta \ubar\eta\bar{C}_1\bar P\]
we have
\[ \bar{C}_1\bar{P}(\ubar\theta\ubar\eta + \zeta\ubar\theta - 
\ubar\eta\ubar\theta + \ubar\eta\ubar{\bar\theta}) = 0 \]
Applications of (\ref{EqRealityPTimesTetradCoeff}) allows us to
replace $+\ubar\eta\ubar{\bar\theta}$ by $-\eta\ubar\theta$ in the
brackets, and so, since
$\ubar\theta\bar{C}_1\bar{P} = \ubar{D}C_1P \neq 0$, we must have $\zeta = \eta$,
which considerably simplifies calculations. Next we write
\[ b^a = -i\frac{y^2+z^2}{\mathfrak{A}-z^2}(\eta C_1P\bar{m}^a -
\bar\eta \bar{C}_1\bar{P}m^a) = i\frac{1}{\mathfrak{A}-z^2} \left( \ubar\eta
C_1\bar{C}_1^2P\bar{P}^2\bar{m}^a - c.c \right)\]
by expanding $\nabla^az$ in tetrad coefficients, and where $c.c.$
denotes complex conjugate. Then, since
$\ubar{D}z = 0$, 
\[ -i(\mathfrak{A}-z^2)[\lbar,b] = \ubar{D}(\ubar\eta
C_1\bar{C}_1P\bar{P}^2)\bar{m}^a - c.c + \ubar\eta C_1\bar{C}_1^2P\bar{P}^2[\ubar{D},\bar\delta] -
c.c\]
We consider the commutator relation, simplified appropriately in view
of computations above and in the proof of Lemma \ref{LemmaNormGradyz},
\[ [\ubar{D},\bar\delta] = -(\Gamma_{213}+\ubar\theta)\bar\delta =
(\Gamma_{123} - \frac{1}{\bar{C}_1\bar{P}})\bar\delta \]
together with the structure equation 
\[ (\ubar{D}+\Gamma_{123})\ubar\eta = \ubar\theta(\eta-\ubar\eta) \]
and the relations in (\ref{EqRealityPTimesTetradCoeff}) and
(\ref{EqNullTetradGradP}), we get
\begin{align*}
\ubar{D}(\ubar\eta C_1\bar{C}_1^2P\bar{P}^2)\bar{m}^a &+ \ubar\eta
C_1\bar{C}_1^2P\bar{P}^2[\ubar{D},\bar\delta] \\
& = (\ubar{D}+\Gamma_{123})\ubar\eta
C_1\bar{C}_1^2P\bar{P}^2\bar{m}^a - \ubar\eta |C_1P|^2\bar{m}^a + \ubar\eta
\ubar{D}(C_1\bar{C}_1^2P\bar{P}^2)\bar{m}^a\\
&= \ubar\theta(\eta-\ubar\eta)C_1\bar{C}_1^2P\bar{P}^2\bar{m}^a - \ubar\eta
|C_1P|^2\bar{m}^a + \ubar\eta(\ubar\theta\bar{C}_1^3\bar{P}^3 + 2\ubar\theta
C_1\bar{C}_1^2P\bar{P}^2)\bar{m}^a\\
&= 0
\end{align*}
Hence $[\lbar,b]=0$. In a similar fashion, we write
\[ n^a = |C_1P|^2l^a +
\frac12(\mathfrak{A} + y^2 + |C_1|^2 - 2 C_1\bar{C}_2y)\lbar^a +
(\mathfrak{A}+y^2)t^a \]
From the fact that $b^a\nabla_ay = 0$ and from the known commutator
relations, we have
\[ [n,b] = [ C_1\bar{C}_1P\bar{P}l,b ] +
\frac12(\mathfrak{A} + y^2 + |C_1|^2 - 2 C_1\bar{C}_2y) [\lbar,b] +
(\mathfrak{A}+y^2)[t,b] \]
of which the second and third terms are already known to vanish.
We evaluate $[C_1\bar{C}_1P\bar{P}l,b]$ in the same way we evaluated $[\lbar,b]$,
and a calculation shows that it also vanishes. To evaluate
$[\lbar,n]$, we need to calculate $[\lbar,l]$. To do so we write
\[ t^a = -U\lbar^a - l^a - \bar\eta C_1 P m^a -
\eta\bar{C}_1\bar{P}\bar{m}^a \]
Since $[\lbar,t] =0$, we infer
\begin{align*}
[\lbar,l] &= -[\lbar,U\lbar + \bar\eta C_1 P m + \eta\bar{C}_1\bar{P}\bar{m} ] \\
&= -\ubar{D}U\lbar - [\lbar, \frac{1}{|C_1P|^2}\bar\eta C_1^2\bar{C}_1P^2\bar{P} m ]
- c.c.
\end{align*}
Notice that in the proof above for $[\lbar,b] = 0$ we have
demonstrated that $[\lbar,\bar\eta C_1^2\bar{C}_1P^2\bar{P} m] = 0$, so 
\[ [\lbar,l] = -\ubar{D}U\lbar + \frac{\ubar{D}(|C_1P|^2)}{|C_1P|^2}
(\bar\eta C_1 P m + \eta\bar{C}_1\bar{P}\bar{m}) \]
Direct computation yields
\[ \ubar{D} U =  \frac{y-C_1\bar{C}_2}{y^2+z^2} -
\frac{2yU}{y^2+z^2}\]
and
\[ \ubar{D}(C_1\bar{C}_1P\bar{P}) = 2y \]
(recall that we set $\ubar{D}y = 1$) so we conclude that
\[ [\lbar,l] = - \frac{y-C_1\bar{C}_2}{y^2+z^2}\lbar -
\frac{2y}{y^2+z^2}(l + t) \]
So, using the decomposition for $n^a$ given above
\begin{align*}
[ \lbar, n ] & = [\lbar , (y^2+z^2) l + (y^2+z^2)U\lbar +
(\mathfrak{A}+y^2) t ]\\
&= 2y l + (y-2C_1\bar{C}_2)\lbar + 2y t +
(y^2+z^2)[\lbar,l] \\
&= 0
\end{align*}

Having checked the commutators, we now calculate the scalar products
between various components. A direct computation from the definition
yields
\begin{align*}
b^2 & = \frac{y^2+z^2}{\mathfrak{A}-z^2} & b\cdot n &= 0 & b\cdot
\lbar &= 0 & b\cdot t &= 0\\
\lbar\cdot n &= \mathfrak{A}-z^2 & \lbar^2 &= 0  & \lbar\cdot t&= 1\\
t\cdot n &= \frac{(|C_1|^2 - 2 C_1\bar{C}_2y)(z^2-\mathfrak{A})}{y^2+z^2} &
t^2 & =-1 - \frac{|C_1|^2-2 C_1\bar{C}_2 y}{y^2+z^2}
\end{align*}
and
\[ n^2 =
(\mathfrak{A}-z^2)\left[ \mathfrak{A}+y^2 - \frac{\mathfrak{A}-z^2}{y^2+z^2}\left( |C_1|^2 - 2
C_1\bar{C}_2y\right) \right] \]
A simple computation shows that the determinant of the matrix of inner
products yields
\[ |\det| = (y^2+z^2)^2 \neq 0 \]
and therefore the vector fields are linearly independent. Thus we have
shown that they form a holonomic basis.

To construct the local isometry to Kerr-Newman space-time, we define
co\"ordinates attached to the holonomic vector fields $t,\lbar,b,n$
with the following rescalings. First, since $\mathfrak{A} > 0$, we
can define $a >0$ such that $\mathfrak{A} = a^2$. Then we can
define the co\"ordinates $r,\theta,V,\phi$ by 
\begin{align*}
t & = \partial_V\\
\lbar &= \partial_r & y &= r\\
b &= \frac{1}{a\sin\theta}\partial_\theta & z &= a\cos\theta
\\
n &= - a\partial_\phi
\end{align*}
Notice that we can define $\theta$ from $z$ in a way that makes sense
since $z^2\leq\mathfrak{A}$. Applying the change of co\"ordinates to
the inner-products above we see that in $r,\theta,V,\phi$ the metric
is identical to the one for the Kerr co\"ordinate of Kerr-Newman
space-time. Furthermore, we see that $n$, or $\partial_\phi$, defines
the corresponding axial Killing vector field. 
\end{proof}

To finish this section, we need to show that the results we obtained
in Propositions \ref{PropLocalIsometryRN} and
\ref{PropLocalIsometryKerrNewman} can be extended to the manifold
$\mathcal{M}$, rather than restricted to
$(\mathcal{M}_\lbar\cup\mathcal{M}_l)$ in the former and 
$(\mathcal{M}_\lbar\cup\mathcal{M}_l)\cap\mathcal{M}_\mathfrak{A}$ in
the latter. We shall need the following lemma (Lemma 6 in
\cite{Mars99}; the lemma and its proof can be carried over to our case
essentially without change, we reproduce them here for completeness)

\begin{lemma}\label{LemmaKillingVanishing}
The vector field $n^a$ is a Killing vector field on the entirety of
$\mathcal{M}$. The set $\mathcal{M}\setminus \mathcal{M}_\mathfrak{A}
= \{ n^a = 0 \}$. Furthermore, 
\begin{itemize}
\item If $\mathfrak{A} = 0$, then
$\mathcal{M}\setminus(\mathcal{M}_\lbar\cup\mathcal{M}_l) = \{ t^a = 0
\}$
\item If $0<\mathfrak{A}\leq (C_1\bar{C}_2)^2 - |C_1|^2$, then
$\mathcal{M}\setminus(\mathcal{M}_\lbar\cup\mathcal{M}_l) = \{ \mbox{
either } n^a-y_+t^a = 0 \mbox{ or } n^a-y_-t^a = 0\}$ where 
\[ y_{\pm} = 2(C_1\bar{C}_2)^2- |C_1|^2\pm 2C_1\bar{C}_2\sqrt{(C_1\bar{C}_2)^2 - |C_1|^2 -
\mathfrak{A}}\]
\item If $\mathfrak{A} > (C_1\bar{C}_2)^2 - |C_1|^2$, then
$\mathcal{M}\setminus(\mathcal{M}_\lbar\cup\mathcal{M}_l) = \emptyset$
\end{itemize}
\end{lemma}
\begin{proof}
First consider the case $\mathfrak{A} = 0$. By Proposition
\ref{PropCharacterHypersurfaceOrthogonal}, we have $z=0$. So the
definition (\ref{EqOriginalDefn}) and (\ref{EqDecompositiont}) show
that $n^a$ vanishes identically. Furthermore, since
$\mathcal{M}_\mathfrak{A} = \emptyset$ in this case, we have that
$n^a$ is a (trivial) Killing vector field on $\mathcal{M}$ vanishing
on $\mathcal{M}\setminus \mathcal{M}_\mathfrak{A}$. It is also clear
from (\ref{EqDecompositiont}) that $t^a = 0 \iff t_al^a = t_a\lbar^a =
0$ in this case, proving the first bullet point. 

Now let $\mathfrak{A}>0$. Then Proposition
\ref{PropLocalIsometryKerrNewman} shows that $n^a$ is Killing on
$(\mathcal{M}_\lbar\cup\mathcal{M}_l)\cap\mathcal{M}_\mathfrak{A}$,
and does not co\"incide with $t^a$. Since $\mathcal{M}_\mathfrak{A}$
is dense in $\mathcal{M}$ (see paragraph immediately before
Proposition \ref{PropLocalIsometryKerrNewman}), we have that $n^a$ is
Killing on $\overline{\mathcal{M}_\lbar\cup\mathcal{M}_l}$ (the
overline denotes set closure). We wish to show that
$\overline{\mathcal{M}_\lbar\cup\mathcal{M}_l} = \mathcal{M}$. Suppose
not, then the open set $\mathcal{U}
=\mathcal{M}\setminus\overline{\mathcal{M}_\lbar\cup\mathcal{M}_l}$ is
non-empty. In $\mathcal{U}$, $t_al^a = t_a\lbar^a = 0$, so by
(\ref{SubeqDecompositionGradyz}), $\nabla^a y = 0$ in $\mathcal{U}$.
Taking the real part of the third identity in Proposition
\ref{PropPIdentities}, we must have $y = C_1\bar{C}_2$ in $\mathcal{U}$, which
by Lemma \ref{LemmaNormGradyz} implies $\mathfrak{A} =
(C_1\bar{C}_2)^2 - |C_1|^2$. Consider the vectorfield defined on all
of $\mathcal{M}$ given by $n^a - (\mathfrak{A} + y^2)t^a = n^a -
[2(C_1\bar{C}_2)^2 - |C_1|^2]t^a$. As it is a
constant coefficient linear combination of non-vanishing independent 
Killing vector fields on
$\overline{\mathcal{M}_\lbar\cup\mathcal{M}_l}$, it is also a
non-vanishing Killing vector field. However, on $\mathcal{U}$, the
vector field vanishes by construction. So we have Killing vector field
on $\mathcal{M}$ that is not identically 0, yet vanishes on an
non-empty open set, which is impossible (see Appendix C.3 in
\cite{Wald}). Therefore $n^a$ is a Killing vector field everywhere on
$\mathcal{M}$. Now, outside of $\mathcal{M}_\mathfrak{A}$, we have
that $z^2 = \mathfrak{A}$ reaches a local maximum, so $\nabla_a z$
must vanish. Therefore from (\ref{EqOriginalDefn}) and
(\ref{EqDecompositiont}) we conclude that $n^a$ vanishes outside
$\mathcal{M}_\mathfrak{A}$ also, proving the second statement in the
lemma. 

For the second a third bullet points, consider the function $U =
\frac12(\nabla y)^2$. By
definition it vanishes outside $\mathcal{M}_\lbar\cup\mathcal{M}_l$.
Using Lemma \ref{LemmaNormGradyz} we see that 
\[ \mathfrak{A} + y^2 + |C_1|^2 - 2C_1\bar{C}_2y = 0 \]
outside $\mathcal{M}_\lbar\cup\mathcal{M}_l$. The two bullet points
are clear in view of the quadratic formula and (\ref{EqOriginalDefn}).
\end{proof}
Now we can complete the main theorem in the same way as \cite{Mars99}.
\begin{proof}[Proof of the Main Theorem]
In view of Propositions \ref{PropLocalIsometryRN} and
\ref{PropLocalIsometryKerrNewman}, we only need to show that the
isometry
thus defined extends to
$\mathcal{M}\setminus(\mathcal{M}_\lbar\cup\mathcal{M}_l)$ in the case
of Reissner-Nordstr\"om and
$\mathcal{M}\setminus[(\mathcal{M}_\lbar\cup\mathcal{M}_l)\cap\mathcal{M}_\mathfrak{A}]$
in the case of Kerr-Newman. Lemma \ref{LemmaKillingVanishing} shows
that those points we are interested in are fixed points of Killing
vector fields, and hence are either isolated points or smooth,
two-dimensional, totally geodesic surfaces. Their complement,
therefore, are connected and dense, with local isometry into the
Kerr-Newman family. Therefore a sufficiently small neighborhood of one
of these fixed-points will have a dense and connected subset isometric
to a patch of Kerr-Newman, whence we can extend to those fixed-points
by continuity.
\end{proof}

\section{Proof of the main global result}\label{SectGlobalRes}

To show Corollary \ref{CorMain}, it suffices to demonstrate that the
global assumption (G) leads to the local assumption (L). 

By asymptotic flatness and the imposed decay rate (the assumption that
the mass and charge at infinity are non-zero), we can assume that
there is a simply connected region $\mathcal{M}_\mathcal{H}$ near
spatial infinity such that $\mathcal{H}^2\neq 0$. It thus suffices to
show that $\mathcal{M}_\mathcal{H} = \mathcal{M}$. Suppose not, then
the former is a proper subset of the latter. Let $p_0\in\mathcal{M}$ be
a point on $\partial\mathcal{M}_\mathcal{H}$. We see that Theorem
\ref{ThmMain} applies to $\mathcal{M}_\mathcal{H}$, with $C_1$ taken
to be $q_E + i q_B$ and $C_3 = M / (q_E - i q_B)$. In particular, the
first equation in Proposition \ref{PropPIdentities} shows that, by
continuity, $t^2 = -1$ at $p_0$. Let $\delta$ be a small neighborhood of
$p_0$ such that $t^a$ is everywhere time-like in $\delta$ with $t^2 <
-\frac14$, then the metric $g$ induces a uniform Riemannian metric on
the bundle of orthogonal subspaces to $t^a$, i.e. $\cup_{p\in\delta}
\{ v\in T_p\mathcal{M} | g(v,t) = 0 \}$. Now, consider a curve
$\gamma: (s_0,1] \to \delta$ such that $\gamma(s) \in
\mathcal{M}_\mathcal{H}$ for $s < 1$, $\gamma(1) = p_0$, and
$\frac{d}{ds}\gamma(s)$ has norm 1 and is orthogonal to $t$. Consider
the function $(q_E + i q_B)P\circ\gamma$. By assumption, $|(q_E + i
q_B)P\circ\gamma |\nearrow
\infty$ as $s\nearrow 1$. Since Lemma \ref{LemmaNormGradyz} guarantees
that $z$ is bounded in $\mathcal{M}_\mathcal{H}$, and hence by
continuity, at $p_0$, we must have that $y$ blows up as we approach
$p_0$
along $\gamma$. However, 
\[ |\frac{d}{ds}(y\circ\gamma)| = |\nabla_{\frac{d}{ds}\gamma}y| \leq
C \sqrt{|\nabla_ay\nabla^ay|} < C' < \infty \]
where the constant $C$ comes from the uniform control on $g$ acting as
a Riemannian metric on the orthogonal subspace to $t^a$ (note that
$t^a\nabla_ay = 0$ since $y$ is a quantity derivable from quantities
that are invariant under the $t$-action), and $C'$ arises because by
Lemma \ref{LemmaNormGradyz}, $\nabla_ay\nabla^ay$ is bounded for all
$|y| > 2M$, which we can guarantee for $s$ sufficiently
close to 1. So we have a contradiction: $y\circ\gamma$ blows up in
finite time while its derivative stays bounded. Therefore
$\mathcal{M}_\mathcal{H} = \mathcal{M}$.

\appendix
\section{Tetrad formalisms}\label{SectTetrad}
The null tetrad formalism of Newman and Penrose is used extensively in
the calculations above, albeit with slightly different notational
conventions. In the following, a dictionary is given between the standard
Newman-Penrose variables (see, e.g. Chapter 7 in
\cite{ExactSolutions}) and the null-structure variables of Ionescu and
Klainerman \cite{IoKl07a} which is used in this paper.

Following Ionescu and Klainerman \cite{IoKl07a}, we consider a
space-time with a natural choice of a \emph{null pair} $\{\lbar,l\}$.
Recall that the complex valued vector field $m$ is said to be
\emph{compatible} with the null pair if
\[ g(l,m) = g(\lbar,m) = g(m,m) = 0~,\quad g(m,\mbar) = 1 \]
where $\mbar$ is the complex conjugate of $m$. Given a null pair, for
any point $p\in\mathcal{M}$, such
a compatible vector field always exist on a sufficiently small
neighborhood of $p$. We say that the vector fields $\{m,\mbar,\lbar,l\}$ form
a \emph{null tetrad} if, in addition, they have positive orientation
$\epsilon_{abcd}m^a\mbar^b\lbar^cl^d = i$ (we can always swap $m$ and
$\mbar$ by the obvious transformation to satisfy this condition). 

The scalar functions corresponding to the connection coefficients
of of the null tetrad are defined, with translation to the
Newman-Penrose formalism, in Table \ref{TableRicciRotation}. The
$\Gamma$-notation is defined by 
\[ \Gamma_{\alpha\beta\gamma} = g(\nabla_{e_\gamma}e_\beta,e_\alpha)
\]
where for $e_1 = m$, $e_2 = \mbar$, $e_3 = \lbar$, and $e_4=l$. It is
clear that $\Gamma_{(\alpha\beta)\gamma} = 0$, i.e. it is
antisymmetric in the first two indices. Two natural\footnote{Buyers
beware: the operations are only natural in so much as those geometric
statements that are agnostic to orientation of the frame vectors.
Indeed, both the under-bar and complex conjugation changes the sign of
the Levi-Civita symbol; while for the complex conjugation it is of
less consequence (since the complex conjugate of $-i$ is $i$, the
sign difference is most naturally absorbed), for the under-bar
operation one needs to take care in application to ascertain that
sign-changes due to, say, the Hodge star operator is not present in
the equation under consideration. In particular, generally
co\"ordinate independent geometric statements (such as the relations
to be developed in this section) will be compatible with consistent
application of the under-bar operations, while statements dependent on a
particular choice of foliation or frame will usually need to be
evaluated on a case-by-case basis.} operations are
then defined: the under-bar (e.g. $\theta\leftrightarrow\ubar\theta$)
corresponds to swapping the indices $3\leftrightarrow 4$ (e.g.
$\Gamma_{142}\leftrightarrow\Gamma_{132}$), and complex
conjugation (e.g. $\theta \leftrightarrow \bar\theta$) corresponds to
swapping the numeric indices $1\leftrightarrow 2$ (e.g. $\Gamma_{142}
\leftrightarrow \Gamma_{241}$). 
\begin{table}
\begin{tabular}{|c|c|c|c|}
\hline
& $\Gamma$-notation & Newman-Penrose & Ionescu-Klainerman\\
\hline
$g(\nabla_\mbar l,m)$ & $\Gamma_{142}$ & $-\rho$ & $\theta$ \\
$g(\nabla_\mbar\lbar,m)$ & $\Gamma_{132}$ & $\bar{\mu}$ & $\ubar\theta$\\
$g(\nabla_m l,m)$ & $\Gamma_{141}$ & $-\sigma$ & $\vartheta$\\
$g(\nabla_m \lbar,m)$ & $\Gamma_{131}$ & $\bar\lambda$ & $\ubar\vartheta$\\
$g(\nabla_l l,m)$ & $\Gamma_{144}$ & $-\kappa$ & $\xi$\\
$g(\nabla_\lbar\lbar,m)$ & $\Gamma_{133}$ & $\bar\nu$ & $\ubar\xi$\\
$g(\nabla_\lbar l,m)$ & $\Gamma_{143}$ & $-\tau$ & $\eta$\\
$g(\nabla_l \lbar,m)$ & $\Gamma_{134}$ & $\bar\pi$ & $\ubar\eta$\\
$g(\nabla_l l,\lbar)$ & $\Gamma_{344}$ & $-2\epsilon + \Gamma_{214}$ &
$\omega$\\
$g(\nabla_\lbar\lbar,l)$ & $\Gamma_{433}$ & $2\gamma + \Gamma_{123}$ &
$\ubar\omega$ \\
$g(\nabla_m l,\lbar)$ & $\Gamma_{341}$ & $-2\beta + \Gamma_{211}$ &
$\zeta = -\ubar\zeta$\\
\hline
\end{tabular}
\caption{\label{TableRicciRotation}Dictionary of Ricci rotation
coefficients vs.~Newman-Penrose spin coefficients
vs.~Ionescu-Klainerman connection coefficients}
\end{table}
We note that
$\theta,\ubar\theta,\vartheta,\ubar\vartheta,\xi,\ubar\xi,\eta,\ubar\eta,\zeta$
are complex-valued, while $\omega$ and $\ubar\omega$ are real-valued;
thus the connection-coefficients defined above, along with
complex-conjugation, defines 20 out of the 24 rotation coefficients:
the only ones not given a ``name'' are
$\Gamma_{121},\Gamma_{122},\Gamma_{123},\Gamma_{124}$, among which the
first two are related by complex-conjugation, and the latter-two by
under-bar.

The directional derivative operators are given by:
\[ D = l^a\nabla_a, \ubar{D} = \lbar^a\nabla_a, \delta = m^a\nabla_a,
\bar\delta = \mbar^a\nabla_a \]
(their respective symbols in Newman-Penrose notation are
$D,\Delta,\delta,\bar\delta$).

The spinor components of the Riemann curvature tensor can be given in
terms of the following: let $W_{abcd}$ be the Weyl curvature tensor,
$S_{ab}$ be the traceless Ricci tensor, and $R$ be the scalar
curvature, we can write
\begin{subequations}
\begin{align}
\Psi_2 &= W(l,m,l,m)\\
\bar\Psi_{-2} = \ubar\Psi_2 &= W(\lbar,m,\lbar,m)\\
\Psi_1 &= W(m,l,\lbar,l)\\
\bar\Psi_{-1} = \ubar\Psi_1 &= W(m,\lbar,l,\lbar)\\
\Psi_0 &= W(\mbar,\lbar,m,l)\\
\Phi_{11} &= S(l,l) \\
\ubar\Phi_{11} &= S(\lbar,\lbar) \\
\Phi_{01} &= S(m,l) \\
\ubar\Phi_{01} &= S(m,\lbar)\\
\Phi_{00} &= S(m,m)\\
\Phi_0 &= \frac12[S(l,\lbar) + S(m,\mbar)] 
\end{align}
\end{subequations}
Notice that the quantities $\Psi_A$, $A\in\{-2,-1,0,1,2\}$ are
automatically anti-self-dual: replacing $W_{abcd} \leftrightarrow
{}^*W_{abcd}$ we have $\Psi_A({}^*W) = (-i)\Psi_A(W)$, this follows
from the orthogonality properties of the null tetrad, as well as the
orientation requirement $\epsilon(m,\mbar,\lbar,l) = i$.  
Using this notation, we can write the \emph{null structure equations},
which are equivalent to the Newman-Penrose equations. We derive them
from the definition of the Riemann curvature tensor:
\[ R_{\alpha\beta\mu\nu} = e_\mu(\Gamma_{\alpha\beta\nu}) -
e_\nu(\Gamma_{\alpha\beta\mu}) +
\Gamma^\rho{}_{\beta\nu}\Gamma_{\alpha\rho\mu} -
\Gamma^\rho{}_{\beta\mu}\Gamma_{\alpha\rho\nu} +
(\Gamma^\rho{}_{\mu\nu} -
\Gamma^\rho{}_{\nu\mu})\Gamma_{\alpha\beta\rho} \]
and that
\[ R_{\alpha\beta\mu\nu} = W_{\alpha\beta\mu\nu} +
\frac12(S_{\alpha\mu}g_{\beta\nu} + S_{\beta\nu}g_{\alpha\mu} -
S_{\alpha\nu}g_{\beta\mu} - S_{\beta\mu}g_{\alpha\nu}) +
\frac{1}{12}R(g_{\alpha\mu}g_{\beta\nu} - g_{\beta\mu}g_{\alpha\nu})
\]
So from $R_{1441} = W_{1441} = -\Psi_2$ we get
\begin{subequations}
\begin{equation}
(D + 2\Gamma_{124})\vartheta - (\delta +\Gamma_{121})\xi = \xi(2\zeta
+ \eta + \ubar\eta) - \vartheta(\omega +\theta +\bar\theta) - \Psi_2
\end{equation}
by taking under-bar of the whole expression, we get for a similar
expression for $R_{1331} = -\ubar\Psi_2$ (in the interest of space, we
omit the obvious changes of variables here). For $R_{1442} =
-\frac12S_{44}$ (and analogously $R_{1332} = -\frac12S_{33}$) we have
\begin{equation}
D\theta - (\bar\delta + \Gamma_{122})\xi = -\theta^2 - \omega\theta
-\vartheta\bar\vartheta + \bar\xi\eta + \xi(2\bar\zeta +
\bar{\ubar\eta}) - \frac12 \Phi_{11}
\end{equation}
From $R_{1443} = -\Psi_1 - \frac12S_{14}$
\begin{equation}
(D+\Gamma_{124})\eta - (\ubar{D}+\Gamma_{123})\xi = -2\ubar\omega\xi +
\theta(\ubar\eta-\eta) + \vartheta(\bar{\ubar\eta} - \bar\eta) -
\Psi_1 - \frac12\Phi_{01}
\end{equation}
From $R_{1431} = \frac12 S_{11}$ we get
\begin{equation}
(\ubar{D} + 2\Gamma_{123})\vartheta - (\delta + \Gamma_{121})\eta =
\eta^2 + \xi\ubar\xi - \theta\ubar\vartheta + \vartheta(\ubar\omega -
\bar{\ubar\theta}) + \frac12 \Phi_{00}
\end{equation}
From $R_{1432} = -\Psi_0 + \frac{1}{12}R$
we have
\begin{equation}
\ubar{D}\theta - (\bar\delta + \Gamma_{122})\eta = \xi\bar{\ubar\xi} +
\eta\bar\eta - \vartheta\bar{\ubar\vartheta} + \theta(\ubar\omega -
\ubar\theta) - \Psi_0 + \frac{R}{12}
\end{equation}
From $R_{1421} = -\Psi_1 + \frac12S_{41}$ we have
\begin{equation}
(\bar\delta + 2\Gamma_{122})\vartheta - \delta\theta = \zeta\theta -
\bar\zeta\vartheta + \eta(\theta-\bar\theta) + \xi(\ubar\theta -
\bar{\ubar\theta}) - \Psi_1 + \frac12\Phi_{01}
\end{equation}
Using $R_{3441} = -\Psi_1 - \frac12S_{41}$ we get
\begin{equation}
(D+\Gamma_{124})\zeta - \delta\omega = \omega(\zeta+\ubar\eta) +
\bar\theta(\ubar\eta - \zeta) + \vartheta(\bar{\ubar\eta} - \bar\zeta)
- \xi(\bar{\ubar\theta}+\ubar\omega) - \bar\xi\ubar\vartheta - \Psi_1
- \frac12\Phi_{01}
\end{equation}
From $R_{3443} = \Psi_0 + \bar{\Psi}_0 -S_{34} + \frac{R}{12}$ we
get
\begin{equation}
D\ubar\omega + \ubar{D}\omega = \bar\xi \ubar\xi + \xi\bar{\ubar\xi}
-\bar\eta\ubar\eta - \eta\bar{\ubar\eta} + \zeta(\bar\eta -
\bar{\ubar\eta}) + \bar\zeta(\eta-\ubar\eta) - (\Psi_0 + \bar\Psi_0) +
\Phi_0 - \frac{R}{12}
\end{equation}
and lastly from $R_{3421} = \Psi_0 - \bar\Psi_0 $ we have
\begin{equation}
(\delta - \Gamma_{121})\bar\zeta - (\bar\delta + \Gamma_{122})\zeta =
(\bar\vartheta\ubar\vartheta - \vartheta\bar{\ubar\vartheta}) +
(\theta\bar{\ubar\theta} - \bar\theta\ubar\theta) +
\ubar\omega(\theta-\bar\theta) - \omega(\ubar\theta -
\bar{\ubar\theta}) - (\Psi_0 - \bar\Psi_0)
\end{equation}
\end{subequations}

In this formalism, we can also write the Maxwell equations: let
\begin{subequations}
\begin{align}
\Upsilon_0 &= \frac12 (H(l,\lbar) + H(\mbar,m)) =
\mathcal{H}_{ab}l^a\lbar^b\\
\Upsilon_1 &= H(l,m)= \mathcal{H}_{ab}l^am^b\\
\bar\Upsilon_{-1} = \ubar\Upsilon_1 &= H(m,\lbar) =
\bar{\mathcal{H}}_{ab}m^a\lbar^b
\end{align}
\end{subequations}
be the spinor components of the Maxwell two-form $H_{ab}$. Maxwell's
equations becomes
\begin{subequations}
\begin{align}
\ubar{D}\Upsilon_0 - (\delta-\Gamma_{121})\Upsilon_{-1} &=
\bar{\ubar\xi}\Upsilon_1 - 2\bar{\ubar\theta}\Upsilon_0 -
(\zeta-\eta)\Upsilon_{-1}\\
(\ubar{D}+\Gamma_{123})\Upsilon_1 - \delta\Upsilon_0 &=
(\ubar\omega-\bar{\ubar\theta})\Upsilon_1 + 2\eta\Upsilon_0 -
\vartheta\Upsilon_{-1}
\end{align}
\end{subequations}
and their under-bar counterparts. 

We also need the Bianchi identities
\[ \nabla_{[e}R_{ab]cd} = 0 \]
Note that this implies
\[ \nabla^eW_{ebcd} = \nabla_{[c}S_{d]b} - \frac{1}{12}g_{b[c}\nabla_{d]}R
=: J_{bcd} \]
which gives
\[ \nabla_{[e}W_{ab]cd} = \frac16\epsilon_{seab}J^{srt}\epsilon_{rtcd}
\]
using the orientation condition $\epsilon(m,\mbar,\lbar,l) = i$ we
calculate
\begin{subequations}
\begin{align}
(\bar\delta + 2\Gamma_{122})\Psi_2 &- (D+\Gamma_{124})\Psi_1 +
\frac12\delta\Phi_{11} - \frac12(D+\Gamma_{124})\Phi_{01} \\
\nonumber & =
-(2\bar\zeta + \bar{\ubar\eta})\Psi_2 + (4\theta+\omega)\Psi_1 +
3\xi\Psi_0 \\ \nonumber & \qquad -
(\bar{\theta}+\frac12\omega)\Phi_{01} -
\vartheta\bar{\Phi}_{01} + (\zeta + \frac12\ubar\eta)\Phi_{11} +
\xi\Phi_0 + \frac12\bar\xi\Phi_{00}\\
(\ubar{D}+2\Gamma_{123})\Psi_2 & - (\delta+\Gamma_{121})\Psi_1 +
\frac12(D+2\Gamma_{124})\Phi_{00} - \frac12(\delta +
\Gamma_{121})\Phi_{01} \\
\nonumber & = (2\ubar\omega - \bar{\ubar\theta})\Psi_2 + (\zeta +
4\eta)\Psi_1 + 3\vartheta\Psi_0 \\
\nonumber & \qquad - \frac12\bar\theta\Phi_{00} - \vartheta\Phi_0 -
\frac12\ubar\vartheta\Phi_{11} +\xi\ubar\Phi_{01} +
(\frac12\zeta+\ubar\eta)\Phi_{01}\\
- (\bar\delta +\Gamma_{122})\Psi_1 & -D\Psi_0 - \frac12 D\Phi_0 +
\frac12(\delta-\Gamma_{121})\bar\Phi_{01} - \frac{1}{24}DR\\
\nonumber & = -\bar{\ubar\vartheta}\Psi_2 +
(2\bar{\ubar\eta}+\bar\zeta)\Psi_1 + 3\theta\Psi_0 +
2\xi\bar{\ubar\Psi}_1 \\
\nonumber & \qquad - \frac12(\zeta+\ubar\eta)\bar\Phi_{01} +
\bar\theta\Phi_0 + \frac12\bar{\ubar\theta}\Phi_{11} +
\frac12\vartheta\bar\Phi_{00} \\
\nonumber & \qquad - \frac12 \bar\xi\ubar\Phi_{01} -
\frac12\bar{\ubar\eta}\Phi_{01} - \frac12 \xi \bar{\ubar\Phi}_{01}\\
(D + \Gamma_{124})\ubar\Psi_1 & + \delta\bar\Psi_0 +
\frac12(\ubar{D}+\Gamma_{123})\Phi_{01} - \frac12\delta\Phi_0 +
\frac{1}{24}\delta R\\
\nonumber & = -2\ubar\vartheta\bar\Psi_1 - 3\ubar\eta\bar\Psi_0 +
(\omega-2\bar\theta)\ubar\Psi_1 + \bar\xi\ubar\Psi_2\\
\nonumber & \qquad + \frac12(\ubar\omega-\bar{\ubar\theta})\Phi_{01}
-\frac12\bar\theta\ubar\Phi_{01} -
\frac12\vartheta\bar{\ubar\Phi}_{01} -
\frac12\ubar\vartheta\bar\Phi_{01} \\
\nonumber & \qquad + \frac12\bar\eta\Phi_{00} +
\eta\Phi_0
\end{align}
In addition, we can also take the trace of the Bianchi identities,
which gives
\[ 0 = \nabla^eW_{ebc}{}^b = J_{bc}{}^b \]
and evaluates to 
\begin{align}
-\delta\Phi_0 &- (\bar\delta+2\Gamma_{122})\Phi_{00} +
(\ubar{D}+\Gamma_{123})\Phi_{01} +
(D+\Gamma_{124})\ubar\Phi_{01} + \frac14 \delta R \\
\nonumber &= (\bar\eta + \bar{\ubar\eta})\Phi_{00} +
2(\eta+\ubar{\eta})\Phi_0 + ( \omega -2\theta -\bar\theta) \ubar\Phi_{01}
+ (\ubar\omega - 2\ubar\theta - \bar{\ubar\theta}) \Phi_{01}\\
\nonumber & \qquad -\vartheta\bar{\ubar\Phi}_{01} -\ubar\vartheta \bar\Phi_{01}
+\xi\ubar\Phi_{11} + \ubar\xi \Phi_{11}\\
D\Phi_0 &+ \ubar{D}\Phi_{11} -(\delta-\Gamma_{121})\bar\Phi_{01} - (\bar\delta+\Gamma_{122})\Phi_{01} + \frac14 DR\\
\nonumber &= -\bar\vartheta\Phi_{00} -2(\bar\theta +
\theta) \Phi_0 + \bar\xi \ubar\Phi_{01} +(\bar\zeta + 2\bar\eta +
\bar{\ubar\eta})\Phi_{01} -\vartheta\bar\Phi_{00}\\
\nonumber & \qquad +\xi \bar{\ubar\Phi}_{01} + (\zeta + 2\eta +
\ubar\eta)\bar\Phi_{01} + (2\ubar\omega -\ubar\theta -
\bar{\ubar\theta}) \Phi_{11}
\end{align}
\end{subequations}
A simple identification using Table \ref{TableRicciRotation} and 
the definitions for various spinor components of the Riemann and
traceless Ricci tensors shows that one can recover all of the Bianchi
identities in Newman-Penrose formalism from the above six equations
through the action of complex-conjugation and under-barring. 

Lastly, to complete the formalism, we record the commutator relations
\begin{subequations}
\begin{align}
[D,\ubar{D}] &=(\ubar\eta - \eta)\bar\delta + (\ubar{\bar\eta} -
\bar\eta)\delta - \ubar\omega D + \omega \ubar{D}\\
[D,\delta] &= - \vartheta\bar\delta - (\Gamma_{124} +
\bar\theta)\delta + (\ubar{\eta}+\zeta)D + \xi\ubar{D}\\
[\delta,\bar{\delta}] &= \Gamma_{121}\bar\delta
+ \Gamma_{122}\delta +(\ubar{\bar\theta}-\ubar\theta)D
+(\bar\theta-\theta)\ubar{D}
\end{align}
\end{subequations}

\bibliographystyle{amsalpha}
\bibliography{KerrNewmanBib.bib}

\newcommand{\etalchar}[1]{$^{#1}$}
\providecommand{\bysame}{\leavevmode\hbox to3em{\hrulefill}\thinspace}
\providecommand{\MR}{\relax\ifhmode\unskip\space\fi MR }
\providecommand{\MRhref}[2]{%
  \href{http://www.ams.org/mathscinet-getitem?mr=#1}{#2}
}
\providecommand{\href}[2]{#2}
\begin{thebibliography}{SKM{\etalchar{+}}02}

\bibitem[BCJM04]{Bini04}
Donato Bini, Christian Cherubini, Robert~T. Jantzen, and Giovanni Miniutti,
  \emph{The {Simon} and {Simon-Mars} tensors for stationary {Einstein-Maxwell}
  fields}, Classical and Quantum Gravity \textbf{21} (2004), 1987--1998.

\bibitem[Bun83]{Bunting83}
Gary Bunting, \emph{Proof of the uniqueness conjecture for black holes}, Ph.D.
  thesis, University of New England, Australia, 1983.

\bibitem[Car73]{CarterLH}
Brandon Carter, \emph{Black hole equilibrium states}, Gordon and Breach, 1973.

\bibitem[CK93]{ChKl93}
Demetrios Christodoulou and Sergiu Klainerman, \emph{The global nonlinear
  stability of the {Minkowski} space}, Princeton University Press, 1993.

\bibitem[DKM84]{DebeverKamranMcLenaghan83}
R.~Debever, N.~Kamran, and R.~G. McLenaghan, \emph{Exhaustive integration and a
  single expression for the general solution of the type {$D$} vacuum and
  electrovac field equations with cosmological constant for a nonsingular
  aligned {Maxwell} field}, Journal of Mathematical Physics \textbf{25} (1984),
  no.~6, 1955--1972.

\bibitem[IK07a]{IoKl07a}
Alexandru~D. Ionescu and Sergiu Klainerman, \emph{On the uniqueness of smooth,
  stationary black holes in vacuum}, preprint {arXiv:0711.0040v1 [gr-qc]}
  (2007).

\bibitem[IK07b]{IoKl07b}
\bysame, \emph{Uniqueness results for ill posed characteristic problems in
  curved space-times}, preprint {arXiv:0711.0042v1 [gr-qc]} (2007).

\bibitem[Isr67]{Israel67}
Werner Israel, \emph{Event horizons in static vacuum space-times}, Physical
  review \textbf{164} (1967), no.~5, 1776--1779.

\bibitem[Isr68]{Israel68}
\bysame, \emph{Event horizons in static electrovac space-times}, Communications
  in mathematical physics \textbf{8} (1968), 245--260.

\bibitem[JR05]{JoRa05}
Nils~Voje Johansen and Finn Ravndal, \emph{On the discovery of {Birkhoff's}
  theorem}, preprint {arXiv:physics/0508163v2 [physics.hist-ph]} (2005).

\bibitem[Ker63]{Kerr63}
Roy~P. Kerr, \emph{Gravitational field of a spinning mass as an example of
  algebraically special metrics}, Physical Review Letters \textbf{11} (1963),
  no.~5, 237--238.

\bibitem[Mar99]{Mars99}
Marc Mars, \emph{A spacetime characterization of the {Kerr} metric}, Classical
  and Quantum Gravity \textbf{16} (1999), 2507--2523.

\bibitem[Mar00]{Mars00}
\bysame, \emph{Uniqueness properties of the {Kerr} metric}, Classical and
  Quantum Gravity \textbf{17} (2000), 3353--3373.

\bibitem[Maz82]{Mazur82}
Pawel~O. Mazur, \emph{Proof of uniqueness of the {Kerr-Newman} black hole
  solution}, Journal of Physics A \textbf{15} (1982), 3173--3180.

\bibitem[NCC{\etalchar{+}}65]{Newman65}
Ezra~T. Newman, E.~Couch, K.~Chinnapared, A.~Exton, A.~Prakash, and
  R.~Torrence, \emph{Metric of a rotating, charged mass}, Journal of
  Mathematical Physics \textbf{6} (1965), no.~6, 918--919.

\bibitem[Rob75]{Robinson75}
D.~C. Robinson, \emph{Uniqueness of the {Kerr} black hole}, Physical Review
  Letters \textbf{34} (1975), no.~14, 905--906.

\bibitem[Sim84a]{Simon84a}
Walter Simon, \emph{Characterizations of the {Kerr} metric}, General Relativity
  and Gravitation \textbf{16} (1984), no.~5, 465--476.

\bibitem[Sim84b]{Simon84b}
\bysame, \emph{The multiple expansion of stationary {Einstein-Maxwell} fields},
  Journal of Mathematical Physics \textbf{25} (1984), no.~4, 1035--1038.

\bibitem[SKM{\etalchar{+}}02]{ExactSolutions}
Hans Stephani, Dietrich Kramer, Malcolm MacCallum, Cornelius Hoenselaers, and
  Eduard Herit, \emph{Exact solutions of {Einstein's} field equations}, second
  ed., Cambridge University Press, 2002.

\bibitem[Wal84]{Wald}
Robert~M. Wald, \emph{General relativity}, University of Chicago Press, 1984.

\end{thebibliography}

\end{document}